\documentclass{article}

\usepackage{changes}

\renewcommand\deleted[1]{}

\usepackage[utf8]{inputenc}

\usepackage{a4wide}
\usepackage{authblk}
\usepackage{amsmath,amssymb,amsthm}
\usepackage{graphicx}
\usepackage{todonotes}
\usepackage{url}
\usepackage{color}

\newtheorem{definition}{Definition}
\newtheorem{theorem}{Theorem}
\newtheorem{proposition}{Proposition}
\newtheorem{lemma}{Lemma}
\newtheorem{corollary}{Corollary}
\theoremstyle{remark} 
\theoremstyle{remark} \newtheorem{claim}{Claim}

\newcommand{\SharpP}{\text{\#P}}
\newcommand{\GapP}{\text{GapP}}
\newcommand{\PP}{\text{PP}}
\newcommand{\FP}{\text{FP}}
\newcommand{\NP}{\text{NP}}
\newcommand{\FNP}{\text{FNP}}
\newcommand{\BPP}{\text{BPP}}
\newcommand{\BQP}{\text{BQP}}

\newcommand{\AWPP}{\text{AWPP}}

\newcommand{\WCP}[1]{\text{\#P}_{\{#1\}}}
\newcommand{\WCPN}{\text{\#P}_{\mathbb{N}}}
\newcommand{\WCPZ}{\text{\#P}_{\mathbb{Z}}}
\newcommand{\WCPQ}{\text{\#P}_{\mathbb{Q}}}
\newcommand{\WCPQs}{\text{\#P}^{\mathrm{(pt)}}_{\mathbb{Q}}}
\newcommand{\WCPR}{\text{\#P}_{\mathbb{R}}}

\newcommand{\WDP}[1]{\text{PP}_{\{#1\}}}
\newcommand{\WDPN}{\text{PP}_{\mathbb{N}}}
\newcommand{\WDPZ}{\text{PP}_{\mathbb{Z}}}
\newcommand{\WDPQ}{\text{PP}_{\mathbb{Q}}}
\newcommand{\WDPQs}{\text{PP}^{\mathrm{(pt)}}_{\mathbb{Q}}}
\newcommand{\WDPR}{\text{PP}_{\mathbb{R}}}

\newcommand{\set}[1]{\mathcal{#1}}
\newcommand{\ps}[1]{#1}

\newcommand{\eq}{\ensuremath{\!=\!}} 

\newcommand\mg{\textcolor{black}}
\newcommand\mgg{\textcolor{black}}

\title{\textbf{A Structured View on Weighted Counting with Relations to Counting, Quantum Computation and Applications}}

\author[1,2]{Cassio P. de Campos\thanks{c.decampos@uu.nl}}

\author[3]{Georgios Stamoulis\thanks{georgios.stamoulis@maastrichtuniversity.nl}}

\author[4]{Dennis Weyland\thanks{dennisweyland@gmail.com}} 

\affil[1]{Queen's University Belfast, Northern Ireland, United Kingdom}
\affil[2]{Department of Information and Computing Sciences, Utrecht University, The Netherlands}
\affil[3]{Department of Data Science \& Knowledge Engineering, Maastricht University, The Netherlands}
\affil[4]{\mg{Google, Z\"{u}rich, Switzerland}}

\begin{document}
\date{}
\maketitle

\begin{abstract}
Weighted  counting problems are a natural generalization of counting problems where a weight is associated with every computational path of polynomial-time non-deterministic Turing machines and the goal is to compute the \emph{sum of the weights of all paths} (instead of just computing the \textit{number} of accepting paths). Useful closure properties and
plenty of applications make weighted counting problems interesting. The general definition of these problems captures even undecidable problems, but it turns out that
obtaining an exponentially small additive approximation is just as hard as solving conventional counting problems. In many cases such an approximation is sufficient and
working with weighted counting problems tends to be very convenient.
We present  a structured view on weighted counting by defining classes that depend on the range of the function that assigns weights to paths and by showing the
relationships between these different classes. These classes constitute generalizations of the usual counting problems. Weighted counting allows us to easily cast a
number of famous results of computational complexity in its terms, especially regarding counting and quantum computation. Moreover, these classes are flexible enough and
capture the complexity of various problems in fields such as probabilistic graphical models and stochastic combinatorial optimization. Using the weighted counting terminology
and our results, we are able to \deleted{greatly} simplify and answer some open questions in those fields.
\end{abstract}

\maketitle
\date{}



\section{Introduction}

Counting problems  play an important role in computer science and can be informally defined as finding the number of solutions (that is, computational paths of polynomial
length that end in an accepting state) of a given combinatorial decision problem. Weighted counting problems are a natural generalization of counting problems: a weight is
associated with every computational path and the goal is to compute the sum of the weights of all computational paths. Weighted counting has numerous applications in a wide
variety of fields, such as quantum computation, stochastic combinatorial optimization, probabilistic graphical models, to name but a few. In many situations it is more natural
and more convenient to use weighted counting problems instead of conventional counting problems. For instance, certain proofs can be stated in a more intuitive
manner by using weighted counting. It also offers additional closure properties which do not hold for the conventional counting classes. Because of that and because the
computational complexity of weighted counting is closely related to that of conventional counting, the former might be preferred to the latter when studying some computational
complexity questions.

In this work we give compact definitions for different classes of weighted counting problems, which differ only in the range of their weights. In particular, we study the
complexity of weighted counting problems using natural, integer, rational and real numbers as weights together with some more restricted cases where weights take values on
finite sets such as $\{0,1\}$ and $\{-1,1\}$. If one allows only positive integers as weights, it is easy to show that problems remain in $\SharpP$, the class of
(non-weighted) counting problems \cite{DBLP:journals/tcs/Valiant79}. If negative values are allowed, even if they are integers, then it is not possible anymore to guarantee
that the result is a non-negative value.  Therefore, these problems are not in $\SharpP$ by definition. In order to study these cases, we adopt the terminology of
$\SharpP$-\emph{equivalence} to indicate that a problem is $\SharpP$-hard under metric reductions and can be affinely reduced to a problem in $\SharpP$. This concept of affine reductions is closely related to weighted reductions~\cite{DBLP:journals/jcss/BulatovDGJJR12}, but it is not approximate preserving~\cite{Dyer2004}. In short, it allows polynomial-time preprocessing of the input and an affine transformation of the output from the function.
Essentially, the difference between $\SharpP$-\emph{equivalence} and previously used $\SharpP$-\emph{completeness} under metric reductions is the relaxation obtained with the affine transformation for the membership in $\SharpP$. Hence, problems which are not actually in $\SharpP$
\mg{for definitional reasons} can still be $\SharpP$-\emph{equivalent}. Using this terminology, we show that problems using weights from the set $\{-1,1\}$, \mg{or} from arbitrary integers, are closely related to $\SharpP$ problems, and also that the decision variants of these weighted counting problems remain in $\PP$ (the class of problems related to the decision of whether the majority of \mg{the computational paths of a non-deterministic polynomial time Turing machine are accepting}) \cite{gill1977computational}. While some of these results are considered to be known by the community, we are not aware of an explicit and precise treatment of all relevant cases.

\mg{For} the weighted counting  problems just described \mg{(that is, problems with integer weights)} \mg{the output size} is always bounded by a polynomial in the input size.
In this case, the corresponding decision variants remain in $\PP$. However, this is not necessarily the case for weighted counting problems when one allows weights to be taken from real or even rational numbers. In spite of that, we show that computing arbitrarily good additive approximations for these problems is a $\SharpP$-equivalent problem.
Regarding the corresponding decision problems, it is shown that the superclass of $\PP$ which allows rational weights (we will later define it as $\WDPQ$) contains the halting problem, and hence strictly contains $\PP$.

We also address the  relations between complexity classes of weighted counting problems and those of quantum computation. We show that weighted counting can be used to
represent the acceptance probability of a quantum Turing machine in a natural way. Based on this observation, it is possible to derive complexity results regarding quantum
computation using \mg{arguably simple} proofs. For instance, we are able to show that $\BQP \subseteq \AWPP$ \cite{DBLP:journals/jcss/FortnowR99} using a short
argumentation ($\BQP$ is the class of decision problems that can be solved by a quantum Turing machine with bounded error probability \cite{bernstein1993quantum} and $\AWPP$ is
a quite restricted gap definable counting class and the best known classical upper bound for $\BQP$ \cite{DBLP:journals/jcss/FennerFK94}). Finally, we give an
intuitive proof of the \mg{relatively recent} result that quantum Turing machines with bounded error can be simulated in a time- and space-efficient way using randomized algorithms
with unbounded error that have access to random access memory \cite{DBLP:journals/toc/MelkebeekW12}.

We conclude the paper with a discussion of the implications of \mg{these} results for problems in other fields such as inferences in probabilistic graphical models and
stochastic combinatorial optimization problems. Using these results and terminology, it is possible to simplify complexity proofs for those problems and to \mg{answer} some
\mg{important} questions, namely, upper bounds for the computational complexity of various computational tasks related to probabilistic graphical models and
two-stage stochastic combinatorial optimization problems.

\section{Counting Classes, Generalizations and Historical Developments}

Our starting point is the definition of conventional counting problems and the corresponding class of decision problems. This definition is then altered to allow for the
summation of weights instead of counting accepting paths. Restrictions on the weights lead to several different classes of problems. We discuss basic properties of these
classes and relate them to conventional counting problems. Additionally, we show that the closure properties for conventional counting problems can be extended to
weighted counting problems. While these properties make it convenient to work with weighted counting problems, the general definition leads to problems which may have an
exponentially long output or even capture undecidable problems. We then show that, from a complexity point of view, approximating weighted counting problems up to an
exponentially small additive error is basically equivalent to solving conventional counting problems. In many cases such an approximation is sufficient and therefore weighted
counting is a convenient tool. We begin with an introduction to related concepts. Since counting classes play a central role in this paper, we devote part of this section to discuss them.

\subsection{Counting Problems and Generalizations}

Problems in the complexity class $\NP$ are related  to the question of whether, for a given input, there exists at least one valid certificate of polynomial size (with respect to the
input size) and which can also be checked in polynomial time (with respect to the input size). This complexity class has been generalized in 1979
\cite{DBLP:journals/tcs/Valiant79, DBLP:journals/siamcomp/Valiant79} in the following way: instead of deciding if there exists at least one valid certificate for a given input,
we want to compute, for a given input, the number of valid certificates. Such problems are called counting problems and are subsumed in the complexity class $\SharpP$. More
formally, we can define counting problems as follows.

\begin{definition}[Counting Problem]
We are given a polynomial $p$ and a polynomial-time predicate $T$. The counting problem associated with $p$ and $T$ is to compute for $x \in \{0,1\}^\star$ the function
\[f(x) = \sum_{u \in \{0,1\}^{p(|x|)}} T(x,u) \text{.}\]
\end{definition}

We also need definitions of related complexity classes (such as $\BPP$, $\PP$, $\BQP$) while studying the relationships of weighted counting and other classes. Their definitions are given below.

\begin{definition}[$\BPP$]
A language $L \subseteq \{0,1\}^\star$ is in the complexity class $\BPP$ if and only if there exists a polynomial $p$ and a polynomial-time predicate $M$ (meaning a polynomial-time Turing machine whose output is either 0 or 1) such that for every $x
\in L$, the fraction of strings $y$ of length $p(|x|)$ which satisfy $M(x,y) = 1$ is greater than or equal to $2/3$; for every $x \notin L$, the fraction of strings $y$ of
length $p(|x|)$ which satisfy $M(x,y) = 1$ is less than or equal to $1/3$.\footnote{The string $y$ can be seen as the output of the random coin flips that the probabilistic Turing machine \mg{performs}.}
\end{definition}

\begin{definition}[$\BQP$]
A language $L \subseteq \{0,1\}^\star$ is in the complexity class $\BQP$ if and only if there exists a polynomial-time uniform family of quantum circuits $\{Q_{n}:n\in {\mathbb {N}}\}$, such that:
\begin{enumerate}
\item For all $n \in \mathbb {N}$, $Q_n$ takes n qubits as input and outputs 1 bit.
\item For all $x \in L$, ${\mathrm  {Pr}}(Q_{{|x|}}(x)=1)\geq {\tfrac  {2}{3}}$.
\item For all $x \notin L$, ${\mathrm {Pr}}(Q_{{|x|}}(x)=0)\geq {\tfrac  {2}{3}}$.
\end{enumerate}
\end{definition}
\noindent
This class is defined for a quantum computer and its natural corresponding class for an ordinary computer (or a Turing machine plus a source of randomness) is $\BPP$. Just like P
and $\BPP$, $\BQP$ is low for itself, which means $\BQP^{\BQP} = \BQP$ (we employ the usual notation for oracles). In fact, $\BQP$ is low for $\PP$, meaning that a $\PP$ machine (see the following definition) obtains no benefit from being able to solve $\BQP$ problems instantly, an indication of the possible difference in power between these related classes. The relation between $\BQP$ and $\NP$ is not known. Adding post-selection to $\BQP$ results in the complexity class PostBQP which is equal to $\PP$ \cite{Aaronson3473}.

The closest decision complexity class to $\SharpP$ is $\PP$, standing for Probabilistic Polynomial time.

\begin{definition}[$\PP$]
A language $L \subseteq \{0,1\}^\star$ is in the complexity class $\PP$
if and only if there exists a polynomial $p$ and a polynomial-time predicate $M$
(meaning a polynomial-time Turing machine whose output is either 0 or 1) such that:
\begin{enumerate}
\item For all $x \in L$, the fraction of strings $y$ of length $p(|x|)$ which satisfy $M(x,y) = 1$ is strictly greater than $1/2$.
\item For all $x \notin L$, the fraction of strings $y$ of length $p(|x|)$ which satisfy $M(x,y) = 1$ is less than or equal to $1/2$.
\end{enumerate}
\end{definition}

\noindent
The string $y$ should be interpreted as a string of bits corresponding to random choices. The terms ``less than or equal" can be changed to ``less than"  and the
threshold 1/2 can be replaced by any fixed rational number in $]0,1[$, without changing the class. $\PP$ is a very powerful class: it contains $\BPP$, $\BQP$ and $\NP$. Later
we give a new and extremely short intuitive proof that $\PP$ contains $\BQP$ (and also the related class $\AWPP$, see Section~\ref{section:quantumComputing}) based on the concept of weighted counting. A polynomial-time Turing machine with unlimited access to a $\PP$ oracle can solve the entire polynomial hierarchy. On the
other hand, $\PP$ is contained in PSPACE (polynomial-space bounded Turing Machines; for instance, polynomial space can solve the {\it majority satisfiability} problem simply by going through each possible assignment). Another concept from complexity theory that we need is the notion of \emph{Functional Problems}.

\begin{definition}[$\FP$]
A binary relation $P(x,y)$ is in the complexity class $\FP$ if and only if for every $x\in\{0,1\}^\star$ there is a deterministic polynomial-time (in $|x|$) algorithm that
finds some $y\in\{0,1\}^\star$ such that $P(x,y)$ holds (or tells that such $y$ does not exist).
\end{definition}

The difference between $\FP$ and P is that problems in P have one-bit yes/no answers, while problems in $\FP$ can have any output that can be computed in polynomial time. Just as P
and $\FP$ are closely related, $\NP$ is closely related to $\FNP$. In a straightforward way, the class of relations $\FP^{\SharpP}$ contains all relations that can be computed in polynomial time with access to a $\SharpP$ oracle. We use the terminology $\FP^{\SharpP[m]}$ to limit the overall number of calls to the oracle to $m$. We note that
$\FP^{\SharpP}$ is a very powerful class, as it was shown in \cite{DBLP:journals/tcs/TodaW92}: every function in $\#\mathrm{PH}$ (the class of functions that count the number of accepting paths of polynomial time-bounded nondeterministic oracle Turing machines with oracle sets from the polynomial hierarchy PH) is  metric  reducible to some function in $\SharpP$, that is, $\#\mathrm{PH} \subseteq \FP^{\SharpP[1]}$.

\subsection{Related Work and Historical Developments}

Since its introduction, the class $\SharpP$ has been very successful in characterizing the difficulty of  counting problems: for this, central is the role of
$\SharpP$-completeness \cite{DBLP:journals/tcs/Valiant79,DBLP:journals/siamcomp/Valiant79} introduced to capture the computational complexity of the \textit{permanent}.
Surprisingly, not only all NP-complete problems have their counting version $\SharpP$-complete (this is true under \textit{parsimonious} reductions) but also many ``easy"
problems such as {\it perfect matching} are $\SharpP$-complete in their counting versions. These problems with easy decision version (called ``hard to count-easy to decide'' in
\cite{DBLP:journals/tcs/DurandHK05}) cannot be $\SharpP$-complete under \textit{parsimonious} reductions, but are $\SharpP$-complete under Cook reductions. See \cite{DBLP:conf/mfcs/PagourtzisZ06}
for some closely related classes such as TotP and \#PE (standing for \#P-easy) and their corresponding structural properties and characterizations and also
\cite{DBLP:conf/icalp/HemaspaandraKW01,DBLP:journals/jcss/SalujaST95} for work on some other closely related subclasses of $\SharpP$ that contain functions with
easy decision variants.

In the classical ``textbook'' proof that computing the permanent of a  matrix $M \in \mathbb{Z}^{n\times n}$ is $\SharpP$-complete~\cite{DBLP:conf/istcs/Ben-DorH93}, first it is
shown a reduction to the computation of the permanent of another matrix $M' \in \mathbb{Z}_{\geq 0}^{n \times n}$ and then to the computation of the permanent of a matrix $M''$ with $0/1$ entries. In other words, a counting problem over arbitrary integers is reduced to a classic 0/1 counting problem. This however requires pre and postprocessing, so
it has been argued that $\SharpP$ might not be the correct class when we are interested in counting. Indeed $\SharpP$ has some disadvantages, including not being closed under many operations. Namely, $\SharpP$ is not (or is not known to be) closed under subtractions or binomial coefficients. A direct implication of
the first case is that computing the permanent of a matrix with arbitrary integer entries is not $\SharpP$ by definition. This was a motivation behind the work of
Fenner, Fortnow and Kurtz \cite{DBLP:journals/jcss/FennerFK94} that introduced and defined the complexity class GapP (intuitively, a class of functions
that introduces ``gaps'' between the number of accepting and rejecting computations of NP machines) as an alternative to $\SharpP$ (see also
\cite{DBLP:journals/iandc/FennerFL96} for a more systematic study of gap definability). The authors claim that this class constitutes a natural alternative for $\SharpP$.
Indeed, many computations (such as the permanent over arbitrary integers) which are outside $\SharpP$ fall now inside GapP. Moreover, GapP is shown to be closed under non-trivial
operations such as subtractions and binomial coefficients (see also~\cite{DBLP:conf/istcs/Beigel97}). Arguably, the success of GapP comes from the work of Beigel,
Reingold and Spielman to show that $\PP$ is closed under the operation of intersection~\cite{DBLP:journals/jcss/BeigelRS95} and also due to the fact that the definition of
GapP helps to simplify very important complexity theoretic results: Toda's famous theorem that the entire polynomial hierarchy is contained in $\mathrm{P}^{\PP}$
\cite{DBLP:journals/siamcomp/Toda91} can be cast in terms of GapP~\cite{DBLP:journals/siamcomp/TodaO92}. The very important result that $\BQP \subseteq
\PP$ by Adleman, DeMarrais and Huang \cite{DBLP:journals/siamcomp/AdlemanDH97} can also be simplified and improved with GapP: in
\cite{DBLP:journals/jcss/FortnowR99}, among others, it has been shown, using GapP definable functions, that $\PP^{\BQP} = \PP$. Moreover, it has been shown that $\BQP \subseteq \AWPP$, which is the best current upper bound for $\BQP$. Regarding the simulation of a $\BQP$ machine by a randomized,
unbounded error machine, the most efficient simulation was given in a very recent article by van Melkebeek and Watson~\cite{DBLP:journals/toc/MelkebeekW12}.
In spite of that, it seems that the class GapP cannot capture the complexity of \textit{weighted counting}, where every computational path has a
weight and we are interested in computing the sum of weights of all computational paths, which is the main focus of this work.

In another line of research, the complexity of counting {\it constraint satisfaction problems} (CSPs)
has been investigated. A CSP can be formulated as follows: Let $D$ be an arbitrary
finite set called the domain set. Let $\mathcal R = \{R_1, R_2, \ldots, R_k\}$ be an arbitrary finite set of relations on $D$. Each $R_i$ has an arity $n_i \in \mathbb{N}$. As
input we have variables ${x_1, \ldots, x_n}$ over ${D}$, and a collection of constraints ${R \in {\cal R}}$, each applied to a subset of \mg{the} variables. The decision query is
whether there is an assignment that satisfies all the constraints and the corresponding counting version asks to compute the number of satisfying assignments. \mg{If we}
identify each $R\in \mathcal{R}$ \mg{with a function $\phi(R)$}, \mg{then} a counting CSP is to evaluate the following so-called partition function on an input instance ${I}$:

\begin{displaymath}
P(I) = \sum_{x_i \in D} \prod_{\phi \in I} f(x_{i_1},\ldots,x_{i_r}),
\end{displaymath}
\noindent
where $\mg{\phi}$ has  arity ${r}$ and is applied to variables ${x_{i_1},\ldots,x_{i_r}}$. If $\mg{\phi}$ is 0/1 valued, then this formulation counts the number of solutions. If $\phi$ can take
arbitrary values, then we have a \textit{weighted} CSP problem (for example, see \cite{DBLP:journals/jcss/BulatovDGJJR12,DBLP:journals/siamcomp/DyerGJ09}).

Feder and Vardi \cite{DBLP:journals/siamcomp/FederV98} conjectured that a CSP over constraints $\mathcal{R}$ is either in P or NP-hard, and very recent work seems to have settled
this question~\cite{bulatov,zhuk}.
In a very recent breakthrough, Cai and Chen \cite{DBLP:journals/jcss/CaiLX14,DBLP:conf/stoc/CaiC12} proved a dichotomy theorem for counting versions of CSPs over
complex weights: Given any finite ${D}$ and any finite set of constraint functions ${{\cal F}}$, where  $f_i: D^{n_i} \rightarrow \mathbb{C}$, the problem of computing the
partition function ${P(I)}$ is either in P or $\SharpP$-hard. On the negative side, the criteria are not known to be computable, but still the dichotomy exists.  Another
similar dichotomy conjecture, but yet to be proven, exists for \#CSP, namely that any \#CSP problem is either in $\FP$ or $\SharpP$-complete. See \cite{DBLP:conf/focs/BulatovD03} for some partial results regarding this conjecture and see also
\cite{DBLP:journals/siamcomp/DyerGJ09,DBLP:conf/coco/CaiCL11,DBLP:journals/jcss/BulatovDGJJR12,DBLP:journals/siamcomp/CaiLX11,DBLP:journals/iandc/Yamakami12,
DBLP:journals/tcs/Yamakami12} for other results in the area of counting versions of CSPs.

Another related result comes from Bl\"aser and Curticapean \cite{DBLP:conf/iwpec/BlaserC12} in which the $\#\textrm{W}[1]$-hardness of the weighted counting of all
$k$-matchings in a bipartite graph was \mg{proven}. There, the weight of a matching is simply the product of the weights of the edges that constitute this matching.

In another direction and more relevant to our study, a weighted counting class $\SharpP_{Q'}$ has been defined by Goldberg and Jerrum~\cite{DBLP:journals/iandc/GoldbergJ08}. This is the class of functions $f: \{0,1\}^* \rightarrow \mathbb{Q}$ that can be written as the division of a $\SharpP$  function \mg{by} an $\FP$ function. This class was used to classify the complexity of approximating the value of Tutte polynomials of a graph that take as argument arbitrary rational numbers. Recall that the Tutte polynomial of a given graph $G$ is a two-variable polynomial $T_G(x,y)$.  Usually, the arguments $x$, $y$ are integers (not necessarily positive), but the authors were interested in the most general case where $x,y$ are arbitrary rational numbers. Tutte
polynomials can encode interesting graph theoretic properties. For instance, $T_G(1,1)$ counts the number of spanning trees in $G$, $T_G(2,1)$ counts the
number of forests in $G$, and they are also closely connected to chromatic polynomials, flow polynomials, etc. The class $\SharpP_{Q'}$ is a strict subclass of the class $\WCPQs$ defined later on. Exactly evaluating Tutte polynomials is $\SharpP$-hard \cite{jaeg90comp} except for some threshold cases. The authors were interested in some dichotomy results and they significantly widened the cases where there exist a fully polynomial-time randomized approximation scheme (FPRAS) for $T_G(x,y)$, and showed that some other particular cases do not attain FPRAS (modulo RP $\neq$ NP).

We close this section by mentioning the result of Yamakami \cite{DBLP:journals/ijfcs/Yamakami03}, probably the closest to the settings in this paper. There,
the author studied a class of quantum functions, namely $\# \mathrm{QP}_{\mathrm{K}}$, which is defined as the set of functions computing the acceptance probability of some
polynomial-time quantum Turing machine, the amplitudes of which are drawn from a set $K$. We follow a similar notation here.
Yamakami also defined the corresponding gap definable functions $\mathrm{GapQP}_{\mathrm{K}} = \# \mathrm{QP}_{\mathrm{K}} - \# \mathrm{QP}_{\mathrm{K}}$ which are named as \textrm{quantum probability gap} functions. Among other very interesting results, Yamakami proved the following characterization of $\PP$ in terms of some generalized quantum classes:

\begin{theorem}[Theorem 7.1 in \cite{DBLP:journals/ijfcs/Yamakami03}]
Let $A \subseteq \{0,1\}^*$ and let $\mathbb{A}$ be the set of complex algebraic numbers. Then, the following statements are all equivalent:
\begin{enumerate}
\item $A \in PP$,
\item $\exists f,g \in \#\mathrm{QP}_{\mathbb{A}}$ such that $\forall x \in \{0,1\}^*$, $x \in A \Leftrightarrow f(x) > g(x)$,
\item $\exists f,g \in \mathrm{GapQP}_{\mathbb{A}}$ such that $\forall x \in \{0,1\}^*$, $x \in A \Leftrightarrow f(x) > g(x)$.
\end{enumerate}
\end{theorem}
\noindent
The usefulness of this result is even more evident when Yamakami, in the same paper, considers the quantum analog of PP, namely $\textrm{PQP}_{\mathrm{K}}$ for some amplitude
set $K$ (where the usual polynomial-time Turing machine of the standard definition of PP is replaced by a certain quantum Turing machine).
Using the above theorem, Yamakami proved that
$\mathrm{PQP}_{\mathbb{A}} = \mathrm{PP}$. Later in this paper, we will show that $\mathrm{PP}$ over arbitrary (approximable) rational numbers contains even \textit{undecidable} problems and thus cannot be equal to PP. However, we do not know if the statement is true if we consider a more restricted definition and we leave this as a very interesting open question.

\subsection{Weighted Counting Problems}

As a straightforward generalization of conventional counting problems, weighted counting problems can be formally defined in the following way.

\begin{definition}[Weighted Counting Problem] \label{definitionWeightedCountingProblem} We are given
  a polynomial $p$ and a function $w: \{0,1\}^\star \times \{0,1\}^\star \rightarrow \mathbb{R}$
  that can be approximated by a polynomial-time (in the size of the first two arguments and the
  third argument) computable function $v: \{0,1\}^\star \times \{0,1\}^\star \times \mathbb{N}
  \rightarrow \mathbb{Z}$, such that $|w(x, u) - v(x, u, b) / 2^b| \leq 2^{-b}$ for all $x \in
  \{0,1\}^\star, u \in \{0,1\}^\star, b \in \mathbb{N}$. The weighted counting problem associated
  with $p$ and $w$ is to compute for $x \in \{0,1\}^\star$ the function \[f(x) = \sum_{u \in
    \{0,1\}^{p(|x|)}} w(x,u) \text{.}\]

\end{definition}
Here, $x$ is the input to the weighted counting function and $w$ represents the $2^{p(|x|)}$ many weights for a given input $x$. With the restriction to $w$ imposed by the approximation property, we limit the range of $w$ to efficiently computable numbers as defined in  \cite{minsky1967computation} (p. 159). \mg{From this we immediately get that any such efficiently computable number is a weighted counting problem.} Similar notions in the context of quantum computation are used in \cite{DBLP:journals/siamcomp/AdlemanDH97, DBLP:journals/toc/MelkebeekW12}. \mg{We will use the notation $\#s$ to represent the numerical value of a binary string $s \in \{0,1\}^*$.}

For a given rational threshold value, we can define the corresponding decision problem in the following way.

\begin{definition}[Weighted Counting Problem, Decision Variant]
  \label{definitionWeightedCountingProblemDecision} We are given a weighted counting problem defined
  by a polynomial $p$ and a function $w: \{0,1\}^\star \times \{0,1\}^\star \rightarrow \mathbb{R}$
  as well as a threshold value $t \in \mathbb{Q}$ (which may depend on the input size $|x|$).
The corresponding decision problem is to decide for $x \in \{0,1\}^\star$ whether $f(x) \geq t$ or not.
\end{definition}

As mentioned earlier, weighted counting problems may have different characteristics depending on the set from which weights are taken.
\begin{definition}
For any given set $S \subseteq \mathbb{R}$, we define the class \#P$_S$ to consist of all weighted counting problems where the range of $w$ is restricted to $S$. The class of the corresponding decision problems is then denoted by $\PP_S$.
\end{definition}

By using this notation, we can define the classes $\WCP{0,1}$, $\WCP{-1,1}$,
$\WCP{-1,0,1}$, $\WCPN$, $\WCPZ$, $\WCPQ$ and $\WCPR$, as well as the corresponding classes of decision problems $\WDP{0,1}$, $\WDP{-1,1}$, $\WDP{-1,0,1}$, $\WDPN$, $\WDPZ$, $\WDPQ$ and $\WDPR$. These classes contain interesting problems, since for instance the {\it permanent} of matrices with entries that are not natural numbers falls in such different classes depending on the range of the entries.
Some inclusions among these classes are trivial, since $S_1\subseteq S_2\Rightarrow (\SharpP_{S_1}\subseteq \SharpP_{S_2})\land (\PP_{S_1}\subseteq \PP_{S_2})$.

It is easy to see that $\WCP{0,1}$ is equal to $\SharpP$ \cite{DBLP:journals/tcs/Valiant79}, the class of (non-weighted) counting problems. The same equality holds for the
corresponding classes of decision problems $\WDP{0,1}$ and $\PP$. Additionally, $\WCP{-1,0,1}$ is equal to $\GapP$ \cite{DBLP:journals/jcss/FennerFK94}, the closure
of $\SharpP$ under subtraction,
while $\WCP{-1,1}$ equals $\GapP$ under {\it normal} Turing machines (the computational tree is completely binary),
and therefore $\WDP{-1,1}$ and $\WDP{-1,0,1}$ are equal to $\PP$, since $\SharpP$ and $\GapP$ constitute the same class of decision problems.

The classes  $\WCPR$ and $\WDPR$ are extremely powerful since the output of $\WCPR$ could potentially be infinitely long.
The following result confirms that not only $\WCPR$ and $\WDPR$ but also $\WCPQ$ and  $\WDPQ$ are extremely powerful classes, since they contain the {\it halting} problem. Let $\langle a,b \rangle$ represent the concatenation of the bitstring representations of $a$ and $b$.

\begin{theorem}
\mg{$\WDPQ$} contains undecidable problems. \label{thalt}
\end{theorem}
\begin{proof}
\mg{Define a weighted counting problem (as in Definition~\ref{definitionWeightedCountingProblem}) as follows: let $p$ be the constant $1$ (so the summation of Definition~\ref{definitionWeightedCountingProblem} has only 2 terms), and given $\langle M,y\rangle$ consisting of a Turing machine $M$ and an input $y$, let $w(\langle M,y\rangle,0)=0$ if $M$ on input $y$ does not terminate, and $w(\langle M,y\rangle,0)=-2^{-t}$ if $M(y)$ terminates after $t$ steps; $w(\langle M,y\rangle,1)=0$ for all $\langle M,y\rangle$). Such a weight function $w(\langle M,y\rangle,u)$ can be approximated by a polynomial time computation $v$ which on input $(\langle M,y\rangle,0,b)$, with $b \in \mathbb{N}$, simulates $M$ on $y$ and if the simulation terminates in $t\leq b$ steps, it outputs $-2^{b-t}$, otherwise $v$ outputs $0$. The function $v/2^b$ equals $w(\langle M,y\rangle,u)$ if $u=1$, or if $u=0$ and $M(y)$ terminates within $b$ steps or does not terminate at all, while $v$ gives the desired approximation if $M(y)$ terminates in $t>b$ steps.
Therefore, we have constructed a problem in $\WCPQ$ such that $M$ does not terminate on input $y$ if and only if $f(\langle M,y\rangle) \geq 0$, and thus we could
decide (the complement of) the Halting Problem in $\WDPQ$.}
 \end{proof}

\begin{corollary}
PP is strictly contained in $\WDPQ$.
\end{corollary}

Theorem~\ref{thalt} suggests that our definition gives too much power to $\WDPQ$ and $\WCPQ$ and $\WDPQ$ is possibly equal to $\WDPR$. Since we are interested in understanding the complexity of practical problems, as we will discuss in Section~\ref{furtherapp}, we decide to work with a restricted version of $\WDPQ$ and $\WCPQ$ from now on.

\begin{definition}
$\WCPQs$ is the subset of $\WCPQ$ for which the weighted counting problems have their function $w$ (as in Definition~\ref{definitionWeightedCountingProblem}) computable in polynomial time (in the size of its input). $\WDPQs$ is the associated decision version.
\end{definition}

In the following we further investigate the relations among the other classes of weighted counting problems and their corresponding decision versions. Before proceeding, we
must define the terminology of $\SharpP$-equivalence that will be used throughout the paper. This definition was necessary because $\SharpP$ is not closed with respect to
polynomial-time computations. This leads to several inconveniences. We make use of an adapted version of weighted reduction~\cite{DBLP:journals/jcss/BulatovDGJJR12}, as suggested in~\cite{2016arXiv161201120G}, but we allow for affine postprocessing of the output from the desired function. For other definitions of metric reductions, see for example~\cite{FALISZEWSKI2009101}.

\begin{definition}
(Adapted from \cite{FALISZEWSKI2009101}.)
A function $f: \{0,1\}^\star\rightarrow\mathbb{R}$ metric reduces to $g: \{0,1\}^\star\rightarrow\mathbb{R}$ if there are polynomial-time computable functions $h_1,h_2$ such that $\forall x\in \{0,1\}^\star: f(x)=h_1(x, g(h_2(x)))$. Problem $A$ metric reduces to problem $B$ if $A$ regards the computation of a function $f$ and $B$ regards the computation of a function $g$, as just defined.
\end{definition}

\begin{definition}\label{defff}
A problem $A$ is $\SharpP$-hard under metric reductions if
$\SharpP\subseteq \FP^{A[1]}$. Problem $A$ is said $\SharpP$-complete if $A$ is $\SharpP$-hard and $A\in\SharpP$.
\end{definition}

Note that we use a restricted version of hardness in Definition~\ref{defff} when compared to the definition in Valiant's seminal work~\cite{DBLP:journals/siamcomp/Valiant79}. There, an unlimited number of oracle calls are allowed, that is, {problem $y$ is $\SharpP$-hard if $\SharpP\subseteq \FP^y$}. This definition with unlimited calls allows $\PP$-complete problems such as {\it majority satisfiability} to be $\SharpP$-complete.
This is somewhat undesired, since the classes $\PP$ and $\SharpP$ are not known to have similar power. For instance, $\PP=\text{P}^{\PP[\log]}$~\cite{DBLP:journals/jcss/BeigelRS95,Fortnow19961}, while the corresponding question for $\SharpP$ is open, to the best of our knowledge. In particular, it is unknown whether $\FP^{\SharpP[1]}$ is strictly contained in $\FP^{\SharpP[2]}$~\cite{Green1995456,Regan92onthe}. As discussed earlier, $\#\mathrm{PH} \subseteq \FP^{\SharpP[1]}$, suggesting that limited calls to $\SharpP$ are more powerful than limited calls to $\PP$, while unlimited calls give us no distinction between them and $\FP^{\PP}=\FP^{\SharpP}$. We note that this distinction may have practical implications. For instance, the recently very active topic of {\it probabilistic logics} has seen some important problems falling into classes such as $\PP^{\text{NP}},\ldots,\PP^{\text{PH}}$~\cite{fabio2,fabio1}. All those classes are
confined between the limited and unlimited $\PP$ calls: $\PP=\text{P}^{\PP[\log]}\subseteq \PP^{\text{NP}}\subseteq \PP^{\text{PH}}\subseteq \text{P}^{\PP}$~\cite{DBLP:journals/siamcomp/Toda91}.

\begin{definition}
A function $f: \{0,1\}^\star\rightarrow\mathbb{R}$ affinely reduces to $g: \{0,1\}^\star\rightarrow\mathbb{R}$ if there are polynomial-time computable functions $h_1,h_2,h_3$ such that $\forall x\in \{0,1\}^\star: f(x)=h_1(x)\cdot g(h_2(x)) + h_3(x)$. Problem $A$ affinely reduces to problem $B$ if $A$ regards the computation of a function $f$ and $B$ regards the computation of a function $g$, as just defined.\label{def:affine}
\end{definition}

\begin{definition}
A problem $A$ is $\SharpP$-equivalent if it is $\SharpP$-hard under metric reductions
and can be affinely reduced to a problem in $\SharpP$.
\end{definition}

The concept of $\SharpP$-equivalence for problems relaxes the concept of $\SharpP$-completeness under  metric reductions
by introducing the affine reduction instead of actual membership. Arguably,
such concept gives a more useful relation between problems that are said to be $\SharpP$-equivalent when compared to $\SharpP$-completeness. Any
$\SharpP$-complete under  metric reduction problem is also $\SharpP$-equivalent (the inverse does not necessarily hold; for instance,
{\it permanent} of integer matrices is $\SharpP$-equivalent but is not $\SharpP$-complete, since it does not belong to $\SharpP$).
Actually, the derivations we have in this paper would work even if we restricted ourselves to affine reductions for showing hardness,
but metric reductions might be a reasonable compromise. \mg{We note that, by definition, any $\SharpP$-equivalent problem is a weighted counting problem, since the affine reduction can be moved inside/outside the summation of weights.}

\begin{definition}
A class $C$ is $\SharpP$-equivalent if any problem in $C$ can be reduced to a problem in $\SharpP$ via an affine reduction and any problem in $\SharpP$ can be reduced to a problem in $C$ via an affine reduction.
\end{definition}

The notion of class equivalence is commutative, and could be applied to any other functional complexity class. It gives a means to represent the strong relationship
between some classes, as we will see later on.
It is worth noting that the decision variants of problems in $\SharpP$-equivalent classes are in $\PP$.

We first show equality of the classes $\WCP{0,1}$ and $\WCPN$, as well as equality of the classes $\WCP{-1,0,1}$ and $\WCPZ$. This implies equality of the
corresponding classes of decision problems, namely of $\WDP{0,1}$ and $\WDPN$, and of $\WDP{-1,1}$, $\WDP{-1,0,1}$ and $\WDPZ$. These results are widely understood as known by
the community, but to the best of our knowledge they have never been explicitly stated. Since the range of the functions in $\WCP{-1,1}$ is not limited to non-negative integers,
as it is the case for the functions in $\WCP{0,1}$, these two classes cannot be equal. Nevertheless, it is possible to show that all these classes are
$\SharpP$-equivalent.

We later focus on weighted counting problems in $\WCPQs$, $\WCPQ$ and $\WCPR$ as well as on their corresponding decision variants. $\WCPQ$ and $\WCPR$ capture undecidable problems, as we have proven. It is immediate to show that $\WCPZ\subseteq\WCPQs$ and $\WDPZ\subseteq\WDPQs$.
Additionally, the size of the output of weighted counting problems belonging to $\WCPQs$ is not necessarily polynomially bounded by the input size. Therefore,
the inclusion of $\WCPZ$ in $\WCPQs$ is strict, and it is
generally not possible to give polynomial reductions from problems in $\WCPQs$ to any of the more restricted classes. As we will see in Section
\ref{section:weightedCounting:approximation}, it is still possible to solve problems in $\WCPR$
using a $\SharpP$-equivalent problem such that we lose only a small
additive approximation error. This result then implies that $\WDPQs$, $\WDPQ$ and $\WDPR$ are
actually PP as long as we focus on problems whose output size is
polynomially bounded in the input size.

We make use of the following property to show equality between $\WCP{0,1}$ and $\WCPN$. Given the weight function $w$ and
its approximation $v$ (Definition~\ref{definitionWeightedCountingProblem}), we can assure that the integer part of the weights can always be encoded using polynomially many bits. We then add this polynomial to the given
polynomial $p$ and construct a new weight function using only weights of $0$ and $1$, such that the overall sum does not change. We formalize these ideas in the proof of the
following theorem.

\begin{theorem}
$\WCP{0,1} = \WCPN$.
\end{theorem}

\begin{proof}
  Since problems in $\WCP{0,1}$ are by definition in $\WCPN$, we only have to show the
  other inclusion. For a weighted counting problem in $\WCPN$ defined by a polynomial $p$ and a
  weight function $w: \{0,1\}^\star \times \{0,1\}^\star \rightarrow \mathbb{N}$, we construct the
  following weighted counting problem in $\WCP{0,1}$. Take the polynomial $p^\prime = p + q$, where
  $q$ is a polynomial bounding the number of bits required to encode the integer weights
  of our original problem. Define $w^\prime: \{0,1\}^\star \times \{0,1\}^\star \rightarrow
  \{0,1\}$ by \[w^\prime(x,u) = \begin{cases} 1 & \text{if } \#u_2 < w(x,u_1) \\ 0 &
    \text{else,}\end{cases}\] where $u_1$ are the first $p(|x|)$ bits of $u$ and $\#u_2$ is the
  number encoded by the last $q(|x|)$ bits of $u$. Since the two functions defined by the original
  weighted counting problem and the newly constructed weighted counting problem are identical, we
  conclude the proof. \end{proof}

\begin{theorem}
\label{firstResult}
$\WCP{-1,0,1} = \WCPZ$.
\end{theorem}
\begin{proof}
It follows from the same arguments as in the proof of the previous theorem, this time applied to positive and negative weights.
\end{proof}

Observe that $\WCP{-1,1}$ corresponds to GapP functions on normal Turing Machines (the computational tree is completely binary) whereas $\WCP{-1,0,1}=\WCPZ$ corresponds to GapP functions. Since GapP strictly contains normal GapP functions, we immediately have $\WCP{-1,1} \subset \WCP{-1,0,1}$.

The classes $\SharpP = \WCP{0,1} = \WCPN$, $\WCP{-1,1}$ and $\WCP{-1,0,1} = \WCPZ$ cannot be equal, since their ranges of output are different. Nevertheless, there are affine reductions from problems of each of these classes to the other classes.

\begin{theorem}
\label{theoremEquivalent} $\WCP{-1,1}$ and $\WCP{-1,0,1}$ (which equals to $\WCPZ$) are $\SharpP$-equivalent.
\end{theorem}

\begin{proof}
Since  $\WCP{-1,1}\subset \WCP{-1,0,1}$, it is enough to show reductions from $\SharpP$ to $\WCP{-1,1}$ and from $\WCP{-1,0,1}$ to $\SharpP$.
Any problem in $\WCP{0,1}$ (which equals $\SharpP$) can be affinely reduced to $\WCP{-1,1}$
by multiplying all its weights by 2 and then subtracting 1. After solving the problem in $\WCP{-1,1}$, the result is retrieved back using the inverse transformation, that is,
dividing the output by 2 and summing $2^{p(|x|)-1}$.
Any problem in $\WCP{-1,0,1}$ can be affinely reduced to $\WCPN$ (which equals $\SharpP$). For that we can simply add a value of $1$ to every weight of the given $\WCP{-1,0,1}$ problem, obtaining a problem in $\WCPN$ with weights in
$\{0,1,2\}$. We can then retrieve the result of the original problem in $\WCP{-1,0,1}$ by solving this new problem and then subtracting $2^{p(|x|)}$ from the resulting value.
\end{proof}

This relation of equivalence clarifies that, at least from a computational viewpoint, problems that are $\SharpP$-equivalent are not harder to solve than those in $\SharpP$, as a single evaluation of a $\SharpP$ function plus an affine transformation gives us the desired result. This fact might not be so surprising, but it seems that the literature still has not taken it for granted. For instance, Chapter 17.3 of~\cite{DBLP:books/daglib/0023084} (p. 347) suggests to the reader an approach to solve {\it permanent} of integer matrices using unlimited calls to a $\SharpP$ oracle (more precisely, a $\#\text{SAT}$ one). It also suggests two calls to solve matrices with entries in $\{-1,0,1\}$. With the weighted counting framework, it becomes clear that a single call is enough, even if entries would be drawn from rational numbers (as we will see later on -- the key is that the output of permanent is always polynomially bounded in the input size), and that the problem is in $\FP^{\SharpP[1]}$, perhaps also in some lower class (any $\SharpP$-equivalent problem is trivially in $\FP^{\SharpP[1]}$).

In the remainder of this section, we show a proof that for weighted counting problems in $\WCPQs$ the size of the output, even encoded as a fraction, is not necessarily bounded
by a polynomial in the input size. This proves that $\WCPQs$ cannot be $\SharpP$-equivalent. 

\begin{theorem} \label{theoremNOTEquivalent} $\WCPQs$ is \textbf{not} $\SharpP$-equivalent and is {\bf not} a subset of $\FP^{\SharpP}$.
\end{theorem}
\begin{proof}
We focus on the following simple problem. According to Definition \ref{definitionWeightedCountingProblem}, take the polynomial $p$ to be the identity function and take the weight function $w: \{0,1\}^\star \times
\{0,1\}^\star \rightarrow \mathbb{Q}$ such that
\[
 w(x,u) = \begin{cases} 1 / (\#u + 1), & \text{if } \#u + 1 \in \mathbb{P} \text{ and } \#u < x \\ 0, & \text{otherwise.} \end{cases}
\]
\noindent
Here $\#u$ is the number represented by the bitstring $u$ and $\mathbb{P}$ is the set of prime numbers.  Note that this function fulfills the approximation properties of
Definition \ref{definitionWeightedCountingProblem}. For a given input $x \in \mathbb{N}$ of size $\lceil\log x\rceil$, the task is to compute the value
\[
f(x) = \sum_{\substack{p \in \mathbb{P} \\ p \leq x}} 1 / p
= \left( \sum_{\substack{p \in \mathbb{P} \\ p \leq x}} \enspace \prod_{\substack{q \in \mathbb{P} \\ q \leq x, q \neq p}} q \right) / \prod_{\substack{p \in \mathbb{P} \\ p \leq x}} p \text{.}
\]
\noindent
Note that this fraction cannot be simplified.  The only numbers that divide the denominator are prime numbers of value at most $x$. Each of these prime numbers divides all
parts of the sum except one and therefore it does not divide the whole sum.

In order to show that the result cannot be represented efficiently with respect to the input size, we show that it is not possible to represent the denominator efficiently with respect
to the input size. For this purpose, we present a lower bound for the product of all prime numbers between $x/e$ and $x$ for sufficiently large $x$. Using this lower bound we
can then bound the whole product appearing in the denominator of the output from below.

\begin{claim}
The number of prime numbers between $x/e$ and $x$ is bounded away from below by $x / (3 \ln x)$ for sufficiently large $x$.\label{claim1}
\end{claim}
\noindent \textit{Proof of Claim.}
Let $\pi(x)$ denote the prime-counting function. Due to \cite{rosser1962approximate}, we have for $x \geq 17$ that

\[\pi(x) > \frac{x}{\ln x} \enspace\enspace \text{ and } \enspace\enspace \pi(x/e) < 1.25506 \frac{x/e}{\ln (x/e)} \text{.}\]

For the number of prime numbers between $x/e$ and $x$ we then have
\begin{eqnarray*}
\pi(x) - \pi(x/e) > \frac{x}{\ln x} - 1.25506 \frac{x/e}{\ln (x/e)}
> \frac{x}{\ln x} - 1.25506 \frac{x/e}{3/4 \cdot \ln x}
> \frac{1}{3} \frac{x}{\ln x}.
\end{eqnarray*}
Hence this claim holds for sufficiently large $x$.

Claim~\ref{claim1} means that the product of all the prime numbers between $x/e$ and $x$ is bounded from below by $(x/e)^{x / (3 \ln x)}$. This value is doubly exponential in
the input size of $\lceil\log x\rceil$ and a representation of this value would require exponentially many bits in the input size. Note that the result computed by a
conventional counting problem or a counting problem using integer weights can always be represented using only polynomially many bits in the input size. Hence, a polynomial
reduction between these classes cannot exist in general.
In conclusion, any problem in a $\SharpP$-equivalent class, or in $\FP^{\SharpP}$, has output size bounded by a polynomial of the input size. $\WCPQs$ contains problems which do not have such bounded output size, so the result follows.
\end{proof}

\subsection{Closure properties}

Like conventional counting and $\GapP$, weighted counting comes with many useful closure properties. By definition, weighted counting is closed under addition and multiplication
with a constant. Additionally, it is closed under uniform sums of exponentially many terms and uniform products of polynomially many terms. Moreover, weighted counting is
closed under finite sums and products. These properties basically follow from the definition of weighted counting and the proofs are analogous to those for conventional
counting and $\GapP$ \cite{DBLP:conf/istcs/Beigel97}. We summarize these properties in the following proposition. \mg{We note that in the following proposition we do not attempt to be exhaustive and to cover all possible cases, that is, there might be other cases not directly captured by the proposition below where $f_i$ is a weighted counting problem. But the cases below are enough for the purposes of this work.}

\begin{proposition}
\label{proposition:closure} \mg{Let $f_1, \ldots, f_k$, $k \in \mathbb{N}$ be weighted counting problems and $c$ be any efficiently computable number (as in Definition \ref{definitionWeightedCountingProblem})}. Then the following functions $g$ can be computed by a weighted counting problem.
\begin{itemize}
 \item[(a)] $g(x) = f_1(x) + c$, where $c$ is a constant.
 \item[(b)] $g(x) = c \cdot f_1(x)$, where $c$ is a constant (and therefore in particular $-f_1(x)$).
 \item[(c)] $g(x) = \sum_{i=1}^{k} f_i(x)$.
 \item[(d)] $g(x) = \prod_{i=1}^{k} f_i(x)$.
 \item[(e)] $g(x) = \sum_{y \in \{0,1\}^{p(|x|)}} f_1(\langle x, y \rangle)$, where $p$ is a polynomial.
 \item[(f)] $g(x) = \prod_{y \in \{0,1\}^{\lceil \log p(|x|) \rceil}} f_1(\langle x, y \rangle)$, where $p$ is a polynomial.
\end{itemize}
\end{proposition}

\begin{proof}
These properties follow from the definition of weighted counting and the proofs are analogous to those for conventional counting and $\GapP$
\cite{DBLP:conf/istcs/Beigel97}. In order to prove that in each case $g(x)$ can be computed by some weighted counting function, we need to specify the functions $w$ and $v$ (serving as the approximation of $w$) in accordance to Definition \ref{definitionWeightedCountingProblem}.
\mg{We will start with cases (c)-(d): since $c$ is an efficiently computable number, and thus a weighted counting problem, case (a) follows directly from case (c) and (b) follows from (d).}

The general proof idea is to write $g(x)$ in such a way that it is obvious that it constitutes a weighted counting problem, according to Definition \ref{definitionWeightedCountingProblem}. In the two summation cases (c) and (e), the \mg{weights and approximation function} will actually be transferred in a \mg{direct} fashion from the original weight \mg{functions} $w_i$ and their approximators $v_i$.

We start with case (c). Let $f_i$, $i \in \{1,\ldots,k\}$, be a weighted counting problem. According to Definition \ref{definitionWeightedCountingProblem}, \mg{each} such $f_i$ is characterized by a polynomial $p_i$, a weight function $w_i: \{0,1\}^* \times\{0,1\}^*\rightarrow \mathbb{R}$ and the polynomial time approximator $v_i: \{0,1\}^* \times \{0,1\}^*\times \mathbb{N} \rightarrow \mathbb{Z}$. We want to show that $g(x) = \sum_{i=1}^k f_i(x)$ is a weighted counting problem. Firstly, we choose a new polynomial $q$ such that $q(|x|) \geq p_i(|x|) + \lceil \log k \rceil$ for each $i$ and let $\kappa = \lceil \log k \rceil$. Trivially, such a polynomial $q$ exists. Now we have that

\begin{eqnarray*}
g(x) & = & \sum_{i=1}^k f_i(x) \\
& = & \sum_{i=1}^k \Bigg( \sum_{u \in\{0,1\}^{p_i(|x|)}} w_i(x,u) \Bigg) \\
& = & \sum_{i=1}^k \Bigg( \sum_{u \in\{0,1\}^{q(|x|) - \kappa}} \left\{ \begin{array}{cl}
w_i(x,u), & \text{if } \#u < 2^{p_i(|x|)} \\
0, & \text{otherwise}
\end{array} \right.    \Bigg) \\
& = &\sum_{\substack{u_1 \in\{0,1\}^{\kappa} \\ u_2 \in \{0,1\}^{q(|x|) - \kappa}}}
\left\{ \begin{array}{cl}
w_{\#u_1}(x,u_2), & \text{if } \#u_1 \in \{1, \dots, k\} \text{ and } \#u_2 < 2^{p_{\#u_1}(|x|)} \\
0, & \text{otherwise}
\end{array} \right.
\end{eqnarray*}

\noindent
According to Definition \ref{definitionWeightedCountingProblem} this is obviously a weighted counting problem, where each term is approximated by using the original approximation function and accuracy as is.

The proof of case (d) follows the same pattern. We write $g(x)$ as
\begin{eqnarray*}
g(x) & = & \prod_{i=1}^k f_i(x) \\
& = & \prod_{i=1}^k \Bigg( \sum_{u \in \{0,1\}^{p_i(|x|)}} w_i(x,u) \Bigg) \\
& = & \sum_{\substack{u_1 \in \{0,1\}^{p_1(|x|)} \\ \vdots \\ u_k \in  \{0,1\}^{p_k(|x|)}}} \Bigg(\prod_{i=1}^k w_i(x,u_i) \Bigg) \\
& = & \sum_{u \in \{0,1\}^{\sum_{i=1}^k p_i(n)}} w(x,u)\, ,
\end{eqnarray*}
\noindent \mg{where
$w(x,u) = \prod_{i=1}^k w_i(x,u_i)$, with $u=\langle u_1, u_2, \dots, u_k \rangle$ and $u_i \in \{0,1\}^{p_i(|x|)}$, $i \in \{1,\ldots,k\}$. Now we will define the approximator function:
\[
v(x,u, b) = v(x,\langle u_1, u_2, \dots, u_k \rangle, b) = \lfloor 2^b \cdot \prod_{i=1}^k \frac{v_i(x,u_i, b+1+k+\xi_i)}{2^{b+1+k+\xi_i}} \rceil = \lfloor 2^b \cdot v'(x,u, b) \rceil\, ,
\]
\noindent where
$\lfloor r \rceil$ means the unique integer $z$ such that $z - 1/2 < r \leq z + 1/2$ (that is, the nearest integer), and
$\xi_i=\sum_{j\in\{1,\ldots,k\} \setminus \{i\}} \text{ib}_j(x,u_j)$, with $\text{ib}_j(x,u_j)=\lceil \log(1+v_j(x,u_j,0))\rceil$, which guarantees that $2^{\text{ib}_j(x,u_j)} \geq w_j(x,u_j)$.}

\mg{Firstly note that the time requirements of the function $v$ are trivially polynomially bounded on $|x|, |u|$ and $b$ ($k$ could be even polynomial in $|x|$). In order to prove that the approximator correctly works, we will first prove that the term $v'(x,u, b)$ (potentially a fractional number) is very close (within $1/2^{b+1}$) to the true weight $w(x,u)$. By definition, each call to the individual function $v_i(x, u_i, b+1+k+\xi_i)$ in $v'(x,u, b)$ can be used to approximate the term $w_i(x,u_i)$ such that
\[
  \frac{v_i(x,u_i,b+1+k+\xi_i)}{2^{b+1+k+\xi_i}} = w_i(x,u_i) + \alpha_i\, ,
\]
\noindent
where $\alpha_i \in[-2^{-(b+1+k\xi_i)}, 2^{-(b+1+k+\xi_i)}]$. Then, we have that
\begin{eqnarray*}
v'(x,u, b) &=& \prod_{i=1}^k \frac{v_i(x,u_i, b+1+k+\xi_i)}{2^{b+1+k+\xi_i}} \\
           &=& \prod_{i=1}^k \Big( w_i(x,u_i) + \alpha_i \Big)\\
           &=&  w(x,u) + \sum_{J_1, J_2} \Bigg( \prod_{j \in J_1 }w_j(x,u_j) \Bigg)  \Bigg( \prod_{i \in J_2} \alpha_i \Bigg)\, ,
\end{eqnarray*}
\noindent
where $J_1,J_2$, with $J_2\neq\emptyset$, form a partition of $\{1,\ldots,k\}$.
Each of these $2^k-1$ terms indexed by $J_1,J_2$ has value in the interval  $[-2^{-(b+1+k)}, 2^{-(b+1+k)}]$.  This is because $J_2$ contains at least one index $i$ and thus at least the factor $\alpha_i$ which suffices to bring the term to such interval, since in this case $|\prod_{j \in J_1 }w_j(x,u_j)| \leq 2^{\xi_i}$. So, the absolute value of their sum is at most $2^k 2^{-(b+1+k)} = 2^{-(b+1)}$. Therefore, $|w(x,u) - v'(x,u, b)| \leq 2^{-(b+1)}$. Finally,  $v(x,u, b)$ can be written as $2^bv'(x,u, b)+\alpha$, with $\alpha\in [-1/2,1/2]$, and thus
\[
  \left|w(x,u)-\frac{v(x,u,b)}{2^b}\right| = \left|w(x,u) - \frac{2^bv'(x,u, b)+\alpha}{2^b}\right|\leq \left|w(x,u)-v'(x,u,b)\right|+\left|\frac{-\alpha}{2^b}\right|\leq 2^{-b}\, .
\]
}

\medskip
Now, we move on to case (e): $g(x) = \sum_{y \in \{0,1\}^{p(|x|)}} f_1(\langle x,y \rangle)$. This case is of particular interest since it involves a summation of exponentially many terms of a single weighted counting problem $f_1$ characterized by its polynomial $p_1$, weight function $w_1$ and approximation function $v_1$, in accordance to Definition \ref{definitionWeightedCountingProblem}. As before, we have that

\begin{eqnarray*}
g(x) & = & \sum_{y \in \{0,1\}^{p(|x|)}} f_1(\langle x,y \rangle) \\
& = & \sum_{y \in \{0,1\}^{p(|x|)}}  \Bigg(\sum_{u \in \{0,1\}^{p_1(|x|+|y|)}} w_1(\langle x,y \rangle, u) \Bigg) \\
& = & \sum_{\substack{ y \in \{0,1\}^{p(|x|)} \\ u \in \{0,1\}^{p_1(|x| + p(|x|))}}} w_1(\langle x,y \rangle, u) = \sum_{u \in \{0,1\}^{p(|x|) + p_1(|x|+p(|x|))}} w(x,u)
\end{eqnarray*}

\noindent where in the above $w(x,u) = w_1(\langle x,u_1 \rangle, u_2)$ with $u_1 = $ first $p(|x|)$ bits of $u$ and $u_2 = $ the remaining $p_1(|x|+p(|x|))$ bits of $u$. The approximation of $w$ follows immediately from the approximation of $w_1$ by $v_1$.

Finally, we move to case (f) where we have a product of polynomially many related counting functions. Let $f_1$ with all properties as in case (e) and we have

\begin{eqnarray*}
g(x) & = & \prod_{y \in \{0,1\}^{\lceil \log p(|x|) \rceil}} f_1(\langle x,y \rangle) \\
& = & \prod_{y \in \{0,1\}^{\lceil \log p(|x|) \rceil}} \Bigg( \sum_{u \in \{0,1\}^{p_1(|x|+|y|)}} w_1(\langle x,y \rangle, u)  \Bigg) \\
& = & \sum_{\substack{u_1 \in \{0,1\}^{p_1(|x|+|y|)} \\ \vdots \\ u_{2^{|y|}} \in \{0,1\}^{p_1(|x|+|y|)}}} \Bigg( \underbrace{\prod_{y \in \{0,1\}^{\lceil \log p(|x|) \rceil}} w_1(\langle x,y \rangle, u_y )}_{w(x,u)} \Bigg)\, ,
\end{eqnarray*}
\noindent with $w(x,u)=w(x,\langle u_1,\ldots,u_{2^{|y|}}\rangle)$. As in case (d), it follows that $g(x)$ is a weighted counting problem. We again need to control the precision (the third argument in the approximator function $v_1$) but identical arguments as in case (d) suffice for our purposes.
\end{proof}

Combining these results we can show that any multivariate polynomial in polynomially many variables (given as uniform weighted counting problems) with a degree bounded by a
polynomial and with coefficients which are uniformly computed by another weighted counting problem is itself a weighted counting problem. The following theorem formalizes this
statement.

\begin{theorem}
\label{theorem:MultivariatePolynomial} Let $f$ and $c$ be functions computed by weighted counting problems and let $q$ and $r$ be polynomials. Define the function $p$ for a
given input $x \in \{0,1\}^n$ as
\[
 p(x) = \sum_{(e_1,\ldots,e_{q(n)}) \in \{0,1,\ldots,r(n)\}^{q(n)}} c(\langle x,e_1,\ldots,e_{q(n)} \rangle) f(\langle x,1 \rangle)^{e_1}  \cdots f(\langle x,q(n) \rangle)^{e_{q(n)}} \text{.}
\]
Then $p$ can be computed by a weighted counting problem.
\end{theorem}

\begin{proof}
This result follows from Proposition \ref{proposition:closure}.
\end{proof}

We make use of this result in the next section, which deals with the approximation of weighted counting problems by conventional counting problems.

\subsection{Approximation}
\label{section:weightedCounting:approximation}

In the following we show that the approximation of general weighted counting problems up to a small additive approximation error is itself a $\SharpP$-equivalent
problem.

\begin{theorem} \label{theorem:Approximation}
\mg{Given $b \in \mathbb{N}$, approximating the output value of problems in  $\WCPR$ up to an additive approximation error of $2^{-b}$ is $\SharpP$-equivalent.}
\end{theorem}

\begin{proof}
\mg{Approximating problems in $\WCPR$ cannot be easier than $\SharpP$, since $\WCP{0,1}\subseteq\WCPR$}. It remains to show that such problems can be affinely reduced to a problem in $\SharpP$.
\mg{Let a weighted counting problem in $\WCPR$ be defined by a polynomial $p$ and a function $w: \{0,1\}^\star \times \{0,1\}^\star \rightarrow \mathbb{R}$.}
  For a given $x\in \{0,1\}^\star$, we can solve such problem using the
  integer weight function given by $v$ where the last parameter of $v$ is fixed to $p(|x|) + b$ with one
  call to a problem in $\WCPZ$, which is $\SharpP$-equivalent due to Theorem
  \ref{theoremEquivalent}, followed by a division by $2^{p(|x|)+b}$, so the reduction is affine (according to Definition~\ref{def:affine}, set $h_1(x)=2^{-p(|x|)-b}$ and $h_3(x)=0$). Let us call the function computed by this weighted counting problem $f^\prime(x)=\sum_{u \in \{0,1\}^{p(|x|)}} v(x,u,p(|x|)+b) / 2^{p(|x|)+b}$. The approximation error of the resulting
  value with respect to the value obtained by the original weighted counting problem satisfies
\begin{eqnarray*}
\left| f(x) - f^\prime(x) \right| & = & \left| \sum_{u \in \{0,1\}^{p(|x|)}} w(x,u) - \sum_{u \in \{0,1\}^{p(|x|)}} \frac{v(x,u,p(|x|)+b)}{2^{p(|x|)+b}} \right| \\
& \leq & \sum_{u \in \{0,1\}^{p(|x|)}} \left| w(x,u) - \frac{v(x,u,p(|x|)+b)}{2^{p(|x|)+b}} \right| \\
& \leq & \sum_{u \in \{0,1\}^{p(|x|)}} 2^{-p(|x|)-b} = 2^{-b} \text{,}\\
\end{eqnarray*}
\noindent which concludes the proof.
\end{proof}

\mg{Although there might be ``easy" weighted counting problems, the above theorem provides an upper bound on the complexity of approximating the output value that is true for any problem in $\text{\#P}_{\mathbb{R}}$.} We can now use this result to show that weighted counting problems, for which the output (even if encoded as a fraction) is bounded polynomially in the input size, are
in $\FP^{\SharpP[1]}$.

\begin{theorem} We are given a weighted counting problem in $\WCPR$ defined by a polynomial $p$ and
  a function $w: \{0,1\}^\star \times \{0,1\}^\star \rightarrow \mathbb{R}$. If the size of the
  output, \mg{possibly encoded as a fraction},  is bounded by a polynomial $q(|x|)$, then the given weighted
  counting problem is in $\FP^{\SharpP[1]}$. \label{theorem:onecallsharpp} \end{theorem}

\begin{proof} Using Theorem \ref{theorem:Approximation}, we compute with a single call to a problem
  that is $\SharpP$-equivalent a value $f^\prime(x)$ with an additive approximation error of at most
  $2^{-2q(|x|)-2}$. We know that the actual value $f(x)$ is within the interval $\left[f^\prime(x) -
    2^{-2q(|x|)-2}, f^\prime(x) + 2^{-2q(|x|)-2}\right]$ of size $2^{-2q(|x|)-1}$. In an interval of
  this size there is only one rational value which can be encoded by at most $q(|x|)$ many bits.
  Using the Euclidean algorithm to efficiently find the continued fraction expansions of the two interval endpoints (\cite{hardy1979introduction}, p. 131), we are able to easily retrieve this rational value, so the problem is in $\FP^{\SharpP[1]}$ (note that a $\SharpP$-equivalent class has the same power as $\SharpP$ if used as oracle).
\end{proof}

A similar result holds for the corresponding decision variant. In the following theorem we show that if the size of the output is bounded polynomially in the input size,
the decision variant of the weighted counting problem is in $\PP$.

\begin{theorem} We are given the decision variant of a weighted counting problem in $\SharpP_{\mathbb{R}}$ defined
  by a polynomial $p$, a function $w: \{0,1\}^\star \times \{0,1\}^\star \rightarrow \mathbb{R}$ and
  a threshold value \mg{$t \in \mathbb{Q}$}. If the size of the output of the corresponding weighted counting problem, \mg{possibly encoded as a fraction}, is
  bounded by a polynomial $q(|x|)$, then the decision problem is in $\PP$. \label{theorem:onecallpp}
\end{theorem}

\begin{proof}
  This proof is also based on Theorem \ref{theorem:Approximation}. As in the previous proof, we show
  that we can transform our problem to a problem with integer weights, such that the result of this
  problem divided by a certain integer gives us a good approximation for our original problem. Here
  we need two additional properties. First, we want to use an approximation from above, that means
  an approximation which is always at least as large as the exact value. In this way the approximate
  result is not smaller than the given threshold value if the exact result is not smaller than the
  given threshold. Second, we set the accuracy of the approximation such that the
  approximate value remains smaller than the given threshold value if the exact value is smaller
  than the threshold value. We can do this efficiently, since there is a certain gap between
  rational values whose denominators are bounded.

\mgg{First of all, without loss of generality, we assume that $t\in\mathbb{N}$ (if $t<0$, we can replace
the original $w$ by $-w$ and the original $t$ by $-t$; after that, if
$t=r'/r''\notin\mathbb{N}$ with $r',r''\in\mathbb{N}$, then
we can replace it by $r'$ and the function $w$ by $w\cdot r''$; these operations do not change the decision problem).}

  Take the function $v: \{0,1\}^\star \times \{0,1\}^\star \times \mathbb{N} \rightarrow
  \mathbb{Z}$ from Definition \ref{definitionWeightedCountingProblem} in order to obtain a one-side
  approximation for the weights. For this purpose, simply compute one additional bit using the two-side approximation given
  by $v$ and shift the result accordingly. In this way, the time bounds are not changed and yet we
  obtain a function $v^\prime: \{0,1\}^\star \times \{0,1\}^\star \times \mathbb{N} \rightarrow
  \mathbb{Z}$ which approximates the weights and satisfies $0 \leq v^\prime(x, u, b) / 2^b - w(x, u)
  \leq 2^{-b}$ for all $x \in \{0,1\}^\star, u \in \{0,1\}^\star, b \in \mathbb{N}$.

  For a given $c \in \mathbb{N}$, we can create the following problem in $\WDPZ$. Take the same
  polynomial $p$ as for the given problem. Additionally, take $v^\prime(\cdot, \cdot, p(|x|) + c)$
  as the integer weight function and $t^\prime = t 2^{p(|x|) + c}$ as the threshold value. Since the
  approximation error of the weights is always to the same side, it is easy to verify the following property:
  if the actual result of the original problem is at least as large as the threshold value $t$, then
  the result of the new problem is at least as large as the new threshold value $t^\prime$.

  It remains to show that if the actual resulting value of the original problem is smaller than $t$,
  then the resulting value of the new problem is also smaller than $t^\prime$. This can be achieved
  by using an appropriate accuracy for the approximation, namely a value of $c = q(|x|) + 1$. The result of the new problem divided by $2^{p(|x|)+c}$ approximates the
  actual value with an one-side error of at most $2^{-c}$. If the actual result of the
  original problem is smaller than $t$, then it differs from $t$ by at least $2^{-c+1}$, since the
  encoding of the denominator of the resulting value is bounded by $q(|x|)$. In that case the approximation of the actual result with an error of at most $2^{-c}$ is certainly smaller than $t$ and therefore the
  result of the new problem is smaller than $t^\prime$.

  This concludes the proof, since we have shown that we can polynomially reduce the original problem
  to a single call of a problem in $\WDPZ = \PP$.
\end{proof}

For problems from $\WCPQs$ to $\WCPR$, these approximate results are probably the best we can hope to obtain.
The approximation results for weighted counting problems have a very interesting application regarding the multivariate polynomials which have been discussed in the last
section. Of course, we can easily show that for a polynomially bounded output the exact evaluation of the multivariate polynomial is $\SharpP$-equivalent and the
corresponding decision variant is in $\PP$. But instead of investigating such polynomials, it is much more interesting to focus on rational functions, that is, ratios of two
such polynomials. The power of rational functions in the context of counting problems has been demonstrated in \cite{DBLP:journals/jcss/BeigelRS95, Fortnow19961}, where the
closure of PP under intersection and under truth-table reductions have been shown. While the (approximate) evaluation of rational functions seems to be more complex than
solving a counting problem, the situation is different for the corresponding decision variant due to the observation that the sign of a rational function is equal to the sign
of the product of numerator and denominator. If, in addition, the numerator and denominator are bounded away from 0 by an exponentially small term, the sign can be decided in PP.
The following theorem formalizes this statement.

\begin{theorem}
\label{theorem:RationalFunction} Let $p$ and $q$ be two multivariate polynomials which can be computed by weighted counting problems, as obtained in Theorem~\ref{theorem:MultivariatePolynomial}. If there exists a polynomial $t$, such that for every input $x \in \{0,1\}^n$ we have $|p(x)| \geq 2^{-t(n)}$ and $|q(x)| \geq 2^{-t(n)}$,
then the problem of deciding the sign of $p/q$ is in $\PP$.
\end{theorem}

\begin{proof}
According to Theorem \ref{theorem:MultivariatePolynomial}, $p$ and $q$ are both weighted counting problems. \mgg{Since the sign of $p/q$ is the same as
the sign of $p\cdot q$, we will decide the latter}. Due to Theorem \ref{theorem:Approximation}, additive approximations
of $p$ and \mgg{of} $q$ with error at most $2^{-t(n)-1}$ \mgg{are} $\SharpP$-equivalent problems \mgg{(and hence are weighted counting problems)}. Moreover, such approximations preserve the signs of $p$ and \mgg{of} $q$ because of the assumptions of this theorem \mgg{that $|p(x)| \geq 2^{-t(n)}$ and $|q(x)| \geq 2^{-t(n)}$. Since approximated $p$ times approximated $q$ is a weighted counting problem (by Proposition~\ref{proposition:closure}(d)) and its sign is the same as the sign of $p\cdot q$, we can use the decision about such sign to conclude the proof (which is in $\PP$ by Theorem \ref{theorem:onecallpp})}.
\end{proof}

\section{The Closure of PP under Intersection and Truth-Table Reductions}

In this section we give short and intuitive proofs of the closure of PP under intersection and truth-table reductions \cite{DBLP:journals/jcss/BeigelRS95, Fortnow19961}. These
proofs show that it is more convenient to work in the more powerful class of weighted counting problems and to perform an approximation at the very end to come back to the
complexity of conventional counting problems. In this way, results about the approximation using rational functions can be directly used and some
technicalities can be avoided. At this point we would like to emphasize that we are using the exact same approaches as in the original proofs \cite{DBLP:journals/jcss/BeigelRS95,
Fortnow19961}, but in the context of weighted counting. A contribution is \mg{an alternative proof which follows a different path to reach the results}.

The basis for these results (as well as for the original results) is Newman's theorem about approximating the absolute value function~\cite{newman1964rational}. Such theorem states that a rational function (ratio of two polynomials) with a degree of $m$
can approximate the absolute value of a number in the interval $[-1,1]$ with an additive error of at most $3e^{-\sqrt{m}}$. Unfortunately such result does not ensure that the polynomials of numerator and denominator of the rational function are parameterized using rational coefficients and integer.

\begin{theorem}
If $n=(\nu+f(\nu))^2$ is a perfect square for some integer $\nu\geq 4$, $f$ is
a non-negative integer polynomial, and $\epsilon=1-\frac{1}{\sqrt{n}}+\frac{1}{2n}$, then
$| |x|-r(x) | < 2^{-\nu}$ holds throughout $[-2^{f(\nu)},2^{f(\nu)}]$, where
\begin{equation}
 r(x) = \frac{r'(x)}{r''(x)} = x\cdot\frac{p(x) - p(-x)}{p(x) + p(-x)}, \text{ and } ~ p(x) = \prod_{k=0}^{n-1} (x+2^{f(\nu)}\cdot \epsilon^k)\, .
\label{eq:newnewman}
\end{equation}
($r(x)$ is a fraction of polynomials with rational coefficients of length polynomial in $n$. \mgg{Moreover, the degree of $r(x)$ is $O(n)$.})\label{thm:newnewman}
\end{theorem}
\begin{proof}
See~\ref{app:a}.
\end{proof}

\begin{theorem}[Theorem 11 in \cite{DBLP:journals/jcss/BeigelRS95}]
The intersection of finitely many languages in $\PP$ is itself in $\PP$.\label{tinter}
\end{theorem}

\begin{proof}
 Let $L_1, L_2, \ldots, L_k$ be languages in $\PP$ and let $L = \cap_{i=1}^k L_i$ denote the intersection of these languages. Then there exist corresponding $\GapP$ functions $f_1,
f_2, \ldots, f_k$ such that for any input $x \in \{0,1\}^n$ we have $f_i(x)\geq 0\iff x\in L_i$ and a polynomial $q$ such that $|f_i(x)| \leq 2^{q(n)}$, with index $i \in \{1,2,\ldots,k\}$.
Note that $f_i(x)<0\iff f_i(x)\leq -1$.
Due to Theorem~\ref{thm:newnewman} (applied using $\nu=q(n)$ and $f(\nu)=\nu$), there exist \mgg{rational} polynomials $r'_i$ and $r''_i$, with $i \in \{1,2,\ldots,k\}$, of degree polynomial in \mgg{$(\nu+f(\nu))^2=q(n)^2$, and hence in $n$}, such that the value of $|f_i|$ is
approximated by \mgg{$(r'_i\circ f_i)/(r''_i\circ f_i)$ to within $2^{-q(n)}$. Hence $|r'_i(f_i(x))\footnote{We use the more compact notation $r'_i(f_i(x))$ for $(r_i' \circ f_i)(x)$. Similar for $r''_i\circ f_i$.}/r''_i(f_i(x))-f_i(x)| \leq 2^{-q(n)}$ if $x\in L_i$ and $r'_i(f_i(x))/r''_i(f_i(x))-f_i(x) \geq 2-2^{-q(n)}$} if $x\notin L_i$. Using these rational functions, we create the following function for deciding the intersection of the given languages:
\begin{eqnarray*}
 f(x) = 1 - \sum_{i=1}^{k} \left(\frac{r'_i(f_i(x))}{r''_i(f_i(x))} - f_i(x)\right)\, .
\end{eqnarray*}
\noindent Note that $f(x) \geq 0 \iff x \in L$. Manipulating $f(x)$, we have $f(x)=r'(x)/r''(x)$, where
\begin{eqnarray*}
r'(x) = \sum_{i=1}^{k}\left(f_i(x) r''_i(f_i(x)) - r'_i(f_i(x))\right) \prod_{j \neq i} r''_j(f_j(x)) + r''(x)  ~\text{ and }~ r''(x) = \prod_{j=1}^k r''_j(f_j(x))\, . \\
\end{eqnarray*}
\mgg{Since each $r''_j$ is a (rational) polynomial, Theorem~\ref{theorem:MultivariatePolynomial} tells us that it is a weighted counting problem. The same is trivially true for the GapP functions $f_j$. This implies that, by Theorem~\ref{theorem:MultivariatePolynomial}, each $r''j(f_j(x))$ is a weighted counting problem, that is, each composite factor is a polynomial in a single variable whose coefficients are computed by the same weighted counting function for $r''_j$ whose existence is guaranteed by the assumptions of Theorem~\ref{theorem:MultivariatePolynomial}. Then, by Proposition~\ref{proposition:closure}(d) we get that $\prod_{j=1}^k r''_j(f_j(x))$ is a weighted counting problem as well. The same reasoning can be applied to each composite function in $r'$, and then multiple applications of Proposition~\ref{proposition:closure} suffice to show that $r'$ is also a weighted counting problem.}

Finally, a rather sloppy lower bound of $2^{-\Theta(kq(n))}$ on the absolute values of $r'(x)$ and $r''(x)$ allows us to apply Theorem \ref{theorem:RationalFunction} (the bound holds by construction from Expression~\eqref{eq:newnewman}), which concludes the proof.
\end{proof}

\mgg{The following argumentation, which uses powerful previous well-known results, gives an alternative proof of the closure of $\PP$. Since it follows a different path to reach such important result, it might be useful to understand other properties of $\PP$ and related classes.}

\begin{theorem}[Theorem 5.4 in \cite{Fortnow19961}]
$\PP$ is closed under polynomial-time truth-table reductions.
\end{theorem}

\begin{proof} \mgg{(This proof is provided in a high-level language since the details are similar to those of previous proofs.)}
We are given a polynomial time algorithm computing the input for the $\PP$ oracle calls, a polynomial time predicate $g$ representing the truth-table and a $\PP$ oracle represented
by the $\GapP$ function $f$. For a given input $x \in \{0,1\}^n$, let $x_1, x_2, \ldots, x_{t(n)}$ be the polynomially many inputs of at most polynomial size for the oracle
calls. Then there exists a polynomial $q$ such that for any input $x \in \{0,1\}^n$ and any index $i \in \{1,2,\ldots,t(n)\}$ we have $|f(x_i)| \leq 2^{q(n)}$. Due to Theorem~\ref{thm:newnewman}, there exist for each $i \in \{1,2,\ldots,t(n)\}$ \mgg{rational} polynomials $r'$ and $r''$ of degree polynomial in $n$ such that the value of $|f|$ is
approximated by \mgg{$(r'\circ f)/(r''\circ f)$} to within $2^{-q(n)}$. As in the proof of Theorem~\ref{tinter}, we can use the function \mgg{$(r'\circ f)/(r''\circ f) - f$} to flatten positive cases to (near) zero while expanding negative cases away from zero.
Additionally, the truth-table predicate $g$ can be written as a polynomial in the results from the oracle calls. This follows from arithmetization properties of deterministic
Turing machines \cite{babai1991arithmetization} (which were also used in the original proof of this theorem \cite{Fortnow19961}). Inserting the rational approximations into this
polynomial results in a multivariate rational function, \mgg{as done in Theorem~\ref{tinter}}. Due to the accuracy of the rational approximation, the sign of this function can be used to decide
membership with respect to the truth-table reduction. Finally, using a similar technique as in the proof of Theorem~\ref{tinter}, all terms can be brought to the same denominator
and a rather sloppy lower bound of $2^{-\Theta(t(n)q(n))}$ on the absolute values of the numerator and denominator allows us to apply Theorem \ref{theorem:RationalFunction}, which concludes the proof.
\end{proof}

\section{Quantum Computing}
\label{section:quantumComputing}

In this section we investigate the relationship between weighted counting and quantum computation. For a very thorough treatment of the latter subject, see \cite{nielsen}. The most common computational model for quantum computations is the quantum Turing machine \cite{bernstein1993quantum}. Such a Turing machine is similar to a probabilistic Turing machine \cite{aaronson}. Instead of (positive) transition probabilities, which preserve the $\ell_1$ norm, quantum Turing machines use transition amplitudes, which are in general arbitrary complex numbers that preserve the $\ell_2$ norm. During the execution of a probabilistic Turing machine, the machine is at any point of time in exactly one state (according to the underlying probability distribution). The situation for quantum Turing machines is different. The machine is in a superposition of different states, whose probabilities are defined as the squares \mg{of the absolute values} of their amplitudes \mg{(since amplitudes are generally complex numbers)}. The exact nature of the state can only be determined through measurements. Such measurements are performed at the end of the computation to reveal the result, but also intermediate measurements are
possible. Measuring certain bits destroys the superposition of states in the following sense. Only states which are compatible with the outcome of the measurement survive and
the amplitudes change accordingly to preserve the $\ell_2$ norm. The acceptance probability of a quantum Turing machine on a given input is the probability with which the final
measurements reveal an accepting configuration. We follow the standard exposition regarding quantum Turing machines and their acceptance probabilities as in \cite{aaronson} (in particular Chapter 10, p. 139--140).

We begin this section by showing that the acceptance probability of a quantum Turing machine can be represented using weighted counting. If there are no intermediate
measurements, then the acceptance probability of the given quantum Turing machine is the sum of the acceptance probabilities for accepting configurations. The acceptance
probability for an accepting configuration is the absolute square of its amplitude and its amplitude is the sum of the amplitudes of computational paths leading to this state.
Finally, we can efficiently compute the amplitude of a given computational path of the quantum Turing machine to any precision we need. The key observation is that we can now
express the overall acceptance probability by summing over all pairs of computational paths. We call two computational paths compatible if they arrive at the same accepting
configuration. The weight for a pair of computational paths that are compatible is the real part of the product of the amplitude of the first computational path and the complex
conjugate of the amplitude of the second computational path. Here we may consider only the real part of that product, since the imaginary parts cancel each other out over the
whole sum. The weight for a pair of non-compatible computational paths is just 0. To allow intermediate measurements, we have to slightly modify the definition of compatible
computational paths, but the underlying idea does not change. In the following we state this result in a more formal way.

\begin{theorem} \label{theoremQuantum} For a given polynomial time quantum Turing machine, the task
  of computing the acceptance probability on any input is a weighted counting problem in $\WCPR$.
\end{theorem}

\begin{proof}
  \mg{Our proof is based on the original ideas of Adleman, DeMarrais, and Huang   \cite{DBLP:journals/siamcomp/AdlemanDH97} but embedded in the weighted counting framework}. We first focus on polynomial time quantum Turing machines without intermediate measurements. After
  that we show how the result can be extended to the more general case.

  Let $M$ be a quantum Turing machine without intermediate measurements and whose computational time
  is bounded from above by a polynomial $t$. For a given input $x$, let us denote the set of possible accepting
  final configurations of $M$ by $C$ and the set of possible computational paths by $P$. Note that
  the cardinality of both sets is at most exponential in the input length and that elements from
  these sets can be efficiently enumerated. The amplitude for a configuration $c \in C$ is denoted
  by $a_c$ and the amplitude for a computational path $p \in P$ is denoted by $a_p$. We can then
  write the acceptance probability of $M$ for the given input $x$ in the following way (we use
  the notation $p_1 \sim p_2$ to indicate that two computational paths $p_1$ and $p_2$ arrive at the
  same accepting configuration).

\begin{eqnarray}
\text{Pr}(\text{$M$ accepts $x$}) & = & \sum_{c \in C} \left| a_c^2 \right| =  \sum_{c \in C} \, \left| \sum_{\substack{\text{$p \in P$} \\ \text{arriving at $c$}}} a_p \right|^2 \nonumber \\
 & = & \sum_{c \in C} \,\,\, \sum_{\substack{\text{$p_1, p_2 \in P$} \\
  \text{both arriving at $c$}}} a_{p_1} \overline{a_{p_2}} \nonumber \\
 & = & \sum_{\substack{\text{$p_1, p_2 \in P$} \\ p_1 \sim p_2}} a_{p_1}
  \overline{a_{p_2}} \,\,\,\,=  \sum_{\substack{\text{$p_1, p_2 \in P$}\\ p_1
  \sim p_2}}  \text{Re}\left(a_{p_1} \overline{a_{p_2}}\right)\, ,
\label{qeq}
\end{eqnarray}
\noindent where $\text{Re}(\cdot)$ is the real part of a number.  To compute the acceptance probability, we have to sum over all pairs of computational paths, as shown in
Expression~\eqref{qeq}. If two computational paths arrive at the same accepting configuration, we call them compatible. The weight for two compatible computational paths is the
real part of the product of the amplitude of the first path with the complex conjugate of the amplitude of the second part. The weight for two non-compatible paths is just $0$.
In order to approximate the product of the amplitudes of two computational paths up to a specified precision, we just have to multiply the transition amplitudes between all the
computational steps of the two paths with a certain precision. That means we can express the acceptance probability of $M$ as a weighted counting problem.  This derivation is straightforward, but to our best knowledge it had not been used before for showing a direct relationship between quantum computation and weighted counting.

Now we still have to show that the result also holds if we allow intermediate measurements. Observe that this case is not necessary because, as it was shown in
\cite{DBLP:journals/siamcomp/BernsteinV97, DBLP:journals/siamcomp/AdlemanDH97}, we can always postpone the intermediate measurements to the end of the computation. For the sake
of completeness, we include the standard (and not very hard) arguments here.

For this purpose we have to adapt our definition of compatible computational paths. We now say that two computational paths are compatible if they arrive at the same accepting
configuration and if they have the same results for intermediate measurements. We capture this extended definition of two compatible computational paths $p_1$ and $p_2$ by the
same notation of $p_1 \sim p_2$. The general formula is then

\begin{displaymath}
\text{Pr}(\text{$M$ accepts $x$}) \,\,\, = \sum_{\substack{\text{$p_1, p_2 \in P$} \\ p_1 \sim p_2}} \text{Re} \left(a_{p_1} \overline{a_{p_2}}\right).
\end{displaymath}

To show that this formula is indeed correct, we postpone the intermediate measurements using the following standard technique. Here instead of an intermediate measurement, a
controlled {\it not} (or simply CNOT) operation is performed on the bits that have to be measured and special ancilla bits. Then at the very end a measurement on the ancilla
bits is performed. We call this quantum Turing machine $M^\prime$. $M$ and $M^\prime$ have the same acceptance probabilities and in particular they accept the same language.
Now our extended definition of compatible computational paths respects the following property: Two computational paths are compatible for $M$ according to our extended
definition if and only if the corresponding two computational paths are compatible for $M^\prime$ according the basic definition. This shows that such result also holds for the
more general case in which intermediate measurements are allowed. \end{proof}

We point out that there are no specific assumptions about the quantum Turing machines of the previous proof. In particular, we do not impose any restrictions on the transition amplitudes or on the measurements used thereby. Although the most general definition of a quantum Turing machine allows arbitrary weights from the complex plane $\mathbb{C}$ for the transition amplitudes, we can assume that these are restricted to the set of rational numbers $\mathbb{Q}$. In the following we use $\BQP_{\mathbb{K}}$ to emphasize that the underlying quantum TM uses amplitudes drawn form the set $\mathbb{K}$. The reason that we can restrict our attention to rational amplitudes is that, as Adleman, DeMarrais, and Huang  proved \cite{DBLP:journals/siamcomp/AdlemanDH97}, $\BQP_{\tilde{\mathbb{C}}} = \BQP_{\mathbb{Q}} \subseteq \PP$ but $\BQP_{\mathbb{C}}$ has non-recursive languages, where $\tilde{\mathbb{C}}$  means the set of polynomial-time approximable complex numbers. The non-recursiveness of $\BQP_\mathbb{C}$ implies that we have no effective way to determine whether the acceptance probability of a bounded-error polynomial-time QTM is greater than (for instance) $2/3$ (see also relevant discussions in \cite{DBLP:journals/ipl/YamakamiY99}). Thus, in the following we can restrict our attention to rational transition amplitudes ($\BQP = \BQP_{\mathbb{Q}}$).

Based on Theorem \ref{theoremQuantum}, we are now able to give extremely short and intuitive proofs of the well-known facts that
$\BQP  \subseteq \PP$ and (the stronger version) that $\BQP  \subseteq \AWPP$.

\begin{theorem}[\mg{Proved in~\cite{DBLP:journals/siamcomp/AdlemanDH97}}]
$\BQP  \subseteq \PP$.
\end{theorem}

\begin{proof} For any given language $L \in \BQP$, there exists a bounded-error quantum Turing
  machine $M$ which decides $L$. We can now apply Theorem \ref{theoremQuantum} to $M$ and obtain a
  weighted counting problem approximating the acceptance probabilities of $M$. The (approximate)
  decision version corresponding to this weighted counting problem is in $\PP$ and can be used to
  decide $L$ due to the gap between the acceptance probabilities for words inside and outside the
  language. \end{proof}

For the second proof we need an alternative definition of the complexity class $\AWPP$. In \cite{fenner2003pp} the following definition for $\AWPP$ is given.

\begin{definition}
\label{theoremAWPP} A language $L\subseteq\{0,1\}^\star$ is in $\AWPP$ if and only if there exists a polynomial $p$ and a $\GapP$ function $g$ such that, for all $x \in \{0,1\}^\star$,
\begin{eqnarray*}
x \in L & \Rightarrow & 1-c \leq g(x) / 2^p \leq 1, \\
x \notin L & \Rightarrow & 0 \leq g(x) / 2^p \leq c,
\end{eqnarray*}
where $p = p(|x|)$ and $c$ is a constant smaller than $1/2$.
\end{definition}

We now give an equivalent definition based on weighted counting.

\begin{theorem}
\label{theoremAWPP2} A language $L\subseteq\{0,1\}^\star$ is in $\AWPP$ if and only if there exists a function $g$ corresponding to a weighted counting problem such that, for all $x \in
\{0,1\}^\star$,
\begin{eqnarray*}
x \in L & \Rightarrow & 1-c \leq g(x) \leq 1, \\
x \notin L & \Rightarrow & 0 \leq g(x) \leq c,
\end{eqnarray*}
where $c$ is a constant smaller than $1/2$.
\end{theorem}

\begin{proof} Due to Definition \ref{theoremAWPP}, it is clear that for a language $L \in \AWPP$
  there exists a function $g$ corresponding to a weighted counting problem with the above
  properties. Now let us assume that the above properties are fulfilled for a function $g$
  corresponding to a weighted counting problem. We can approximate the weighted counting problem
  using a $\GapP$ function and a power of $2$ as a scaling factor. For a sufficiently good
  approximation we are then able to fulfill the requirements of Definition \ref{theoremAWPP} with
  some constant smaller than $1/2$. \end{proof}

This definition can now be used for \mg{an alternative} proof of $\BQP \subseteq \AWPP$.

\begin{theorem}[\mg{Proved in \cite{DBLP:journals/jcss/FortnowR99}}]
$\BQP \subseteq \AWPP$.
\end{theorem}

\begin{proof} For any given language $L \in \BQP$, there exists a bounded-error quantum Turing
  machine $M$ which decides $L$. By applying Theorem \ref{theoremQuantum} to $M$, we obtain a weighted
  counting problem approximating the acceptance probability of $M$. Since $M$ has a bounded-error
  probability, we can show $L \in \AWPP$ due to Theorem \ref{theoremAWPP2}. \end{proof}

Using the framework of weighted counting, it might be also possible to give proofs for other results in the field of quantum computation. As an example, we give a
shorter proof for a recent result \cite{DBLP:journals/toc/MelkebeekW12} in the remaining part of this section.
In \cite{DBLP:journals/toc/MelkebeekW12}, it is shown among other results that quantum Turing machines with bounded error can be simulated in a time- and space-efficient way
using randomized algorithms with unbounded error that have access to random access memory. The most important version of this result has been stated in the following way.

\begin{theorem}[Theorem 1.1 in \cite{DBLP:journals/toc/MelkebeekW12}] Every
  language solvable by a bounded-error quantum algorithm running in time $t \geq \log n$ and space
  $s \geq \log n$ with algebraic transition amplitudes is also solvable by an unbounded-error
  randomized algorithm with random access memory running in time $\mathcal{O}(t \log t)$ and space $\mathcal{O}(s + \log t)$,
  provided $t$ and $s$ are constructible by a deterministic algorithm with the latter time and space
  bounds. \end{theorem} \begin{proof}
  Let $M$ be a given quantum Turing machine. We use again the formulation of the acceptance
  probability as a weighted counting problem:

\begin{eqnarray*}
 \text{Pr}(\text{$M$ accepts $x$}) \,\,\, = \sum_{\substack{\text{$p_1, p_2 \in P$} \\p_1 \sim p_2}} \text{Re}\left(a_{p_1} \overline{a_{p_2}}\right) \text{.}
\end{eqnarray*}

Based on this formulation, we can build a randomized algorithm with unbounded error: the problem of computing the acceptance probability of $M$ is in $\WCPR$. The corresponding
decision version is in $\WDPR$ and a proper approximation resides in $\PP$, which is a randomized algorithm with unbounded error. To show the desired time and space bounds, we
have to bound the running time and the space used for the computation of each of the summands. For this purpose, one can use approximations of the constant number of transition
amplitudes of $M$ with a precision of $\mathcal{O}(\log t)$ bits. According to \cite{bernstein1993quantum}, the resulting quantum Turing machine is a good approximation to $M$.
Note that the result was originally stated for quantum Turing machines without intermediate measurements, but with a similar argument as in the proof of Theorem
\ref{theoremQuantum}, this result can be extended to the more general case. We now continue our investigations with this machine. We basically follow the proof of
\cite{DBLP:journals/toc/MelkebeekW12}, but we employ the new terminology, which can help to state some parts in a simpler way.

In a preprocessing step, we compute the (constant number of) transition amplitudes of $M$ up to a precision of $\mathcal{O}(\log t)$ bits. This can be done for algebraic
transition amplitudes in a computational time of $\mathcal{O}(\text{polylog}(t))$ and space $\mathcal{O}(\log t)$. Due to the random access of the randomized machine, we can
retrieve these values efficiently whenever they are required. Now the idea is to simulate all pairs of computational paths in parallel, step by step. Again due to the random
access of the randomized machine, we can jump between the two computational paths without an additional overhead. For the moment we additionally keep track of the transition
amplitudes of the paths. This causes some overhead that would invalidate the overall results, but we demonstrate how to avoid such overhead later. Whenever a quantum
measurement occurs, we check if the two computational paths would result in the same measurement. If this is indeed the case, then we continue with the simulation, otherwise we
stop and assign to this path a weight of $0$. At the end, we check if both paths have arrived at the same accepting configuration. If that is the case, then we assign the real
part of the product of the amplitude of the first path with the complex conjugate of the amplitude of the second part as the weight, otherwise we use a weight of $0$. The
resulting randomized algorithm shows the desired behavior, but does not respect the required time and space bounds due to the overhead needed to keep track of the transition
amplitudes of the two computational paths. In the following we show how to fix this issue.

In order to avoid the overhead needed to keep track of the transition amplitudes of the two computational paths, one must realize that it is not even necessary to keep track of
the weights throughout the whole computation. Instead, we may split the amplitudes in positive/negative and real/imaginary parts and perform random branchings at every
computational step, instead of multiplying with the transition amplitudes. For that to work, the resulting weights of all the paths have to be adapted accordingly. If we look
at time and space requirements of this approach, we see that, at each computational step, we have to generate $\mathcal{O}(\log t)$ random bits for the branching, which results
in a computational time of $\mathcal{O}(t \log t)$ and a space of $\mathcal{O}(s + \log t)$, as desired. \end{proof}

We shall emphasize that these bounds rely heavily on the fact that random access to memory is granted to the randomized algorithm. Otherwise, we would obtain slightly weaker
bounds for the general case. With a more elaborate approach it would still be possible to obtain the same bounds without random access to memory if the number of quantum
measurements is $\mathcal{O}((\log t)^2)$. In this case, we have to move the pre-computed transition amplitudes during the simulation process to be always able to access them
efficiently. This approach does not yield any overhead. Additionally, we simulate each of the computational paths until $\mathcal{O}(\log t)$ measurements occur. We must keep
track of the measurement results and switch to the other path. In total, there are at most $\mathcal{O}(\log t)$ switches which cause a time overhead of $\mathcal{O}(s)$, which
can be bounded from above by $\mathcal{O}(t)$.

\section{Further Applications of Weighted Counting}\label{furtherapp}

In this section we take the previous discussions and results into two very relevant problems. First we discuss on inferences in the so called probabilistic graphical models,
which appear in abundance in artificial intelligence, data mining and machine learning. Then we talk about stochastic combinatorial optimization problems, which represent a
very important class of problems in operations research, with applications in numerous fields. Our approach focuses on using the complexity results presented so far to simplify
or even to prove new complexity results for these problems.

\subsection{Probabilistic Graphical Models}

We present an application of the previously discussed complexity results to prove that {\it predictive inferences in probabilistic graphical models} (and especially in Bayesian
networks \cite{koller2009,pearl1988}) is $\SharpP$-equivalent (predictive inferences are also called {\it belief updating} in this context). Historically, the community working
with probabilistic graphical models has been used to cite papers that only partially resolve this question \cite{littman1998,littman2001,roth1996}. The most
cited work is due to Roth \cite{roth1996}, where hardness for $\SharpP$ is demonstrated, but membership is only superficially discussed and no formal proof is given.  This
issue is acknowledged many years later by Kwisthout \cite{kwisthout2011}, who in an attempt to close this question proves that predictive inference is in the so-called
$\SharpP$ {\it modulo normalization} class. This situation is indeed inevitable, because the output here should be a probability value, so the problem cannot be in $\SharpP$.
However, using our terminology and results, we can state that such problems are $\SharpP$-equivalent.
By applying the results of this paper, we obtain a simple alternative proof to (and yet more powerful than) \cite{kwisthout2011}, as we discuss in the sequel of this
section. We start by presenting a general definition of probabilistic graphical models. This definition encompasses
the two most important models: Bayesian networks and Markov random fields, also known as Markov networks (see~\cite{koller2009}, Chapters 3 and 4, respectively).

\begin{definition}[Probabilistic Graphical Model (PGM)]
Let $\set{J}=\{1,\ldots,n\}$ and $X_{\set{J}}\eq (X_1,\dotsc,X_n)$ be a vector of discrete random variables, $\set{J}_1, \dotsc, \set{J}_m$ be a collection of index sets
satisfying $\set{J}_1 \cup \dotsi \cup \set{J}_m\eq \set{J}$, and $\set{P}\eq\{ \phi_1, \dotsc, \phi_m \}$ be a set of (explicitly represented, for instance in a table)
functions over vectors ${X}_{\set{J}_1}, \dotsc,
{X}_{\set{J}_m}$ to non-negative rational numbers, respectively. We call $\set{P}$ a \emph{probabilistic graphical model} for ${X}_{\set{J}}$ if the functions in $\set{P}$
specify a joint probability distribution over assignments ${x}_{\set{J}} \in \ps{{X}_{\set{J}}}$ by
\begin{equation}
\Pr({X}_{\set{J}}\eq {x}_{\set{J}}) = \frac{1}{Z} \prod_{i=1,\ldots,m}
\phi_i({x}_{\set{J}_i}) \, ,
\label{probdef}
\end{equation}
where $Z\eq\sum_{{x}_{\set{J}}\in {X}_{\set{J}}} \prod_{i=1,\ldots,m} \phi_i({x}_{\set{J}_i})$ is a normalizing value known as the \emph{partition
  function} and where the values ${x}_{\set{J}_i}$ are compatible with ${x}_{\set{J}}$. If the functions in $\set{P}$ are not explicitly represented but can be computed
in polynomial time in the size of the input, then we call $\set{P}$ a \emph{generalized probabilistic graphical model}.
\end{definition}

For instance,
a Bayesian network is a probabilistic graphical model as just defined that satisfies the following properties:
\begin{enumerate}
\item[(i)] it has exactly one $\phi_i$ for each $X_i$;
\item[(ii)] $\set{J}_i$ must be such that $i\in \set{J}_i$, and such that $j\notin \set{J}_i$ whenever $j>i$;
\item[(iii)] $\sum_{x_i\in X_i} \phi_i({x}_{\set{J}_i}) = 1$.
\end{enumerate}
Such restrictions naturally
imply $Z=1$ and induce conditional stochastic independence among variables $X_{\set{J}}$ of the domain.

The {\it predictive inference} task can be succinctly defined as follows.
\begin{definition}[Predictive inference in PGMs]
Given a probabilistic graphical model  $\set{P}$ defined by functions $\{ \phi_1, \dotsc, \phi_m \}$ over discrete random variables
$X_{\set{J}}\eq (X_1,\dotsc,X_n)$, {\it predictive inference} (also known as {\it the computation of the partition function})
is the task of computing $Z\eq\sum_{{x}_{\set{J}}\in {X}_{\set{J}}} \prod_{i=1,\ldots,m} \phi_i({x}_{\set{J}_i})$.
\end{definition}

Predictive inference turns out to be a general task that can be used to compute the probability value $\Pr({X}_\set{J'}\eq {x}_\set{J'})$ for any event ${x}_\set{J'}$ of interest: One has simply to
take their specification of the Bayesian network and include into it the indicator functions $\prod_{j\in\set{J'}}\psi_{x_j}(X_j={x_j})$.
It is not hard to check that $Z$ equals to $\Pr({X}_\set{J'}\eq {x}_\set{J'})$ for a Bayesian network that is extended with these new indicator functions. Because of that, we
call this probabilistic graphical model where $Z$ equals to a probability value of interest as {\it queried Bayesian network}.

The complexity of predictive inference depends on how the input is encoded. When all functions in $\set{P}$ are given by numbers directly encoded in the input, it is easy to
show that the output has size that is polynomial in the input size (one could multiply all rational numbers in the input by their least common denominator in order to achieve
an input defined solely by integers -- this is done in polynomial time because all numbers are explicitly given). Hence the following theorem holds for any probabilistic
graphical model, including queried Bayesian networks.
\begin{theorem}
  Given a probabilistic graphical model defined by $\set{P}$, the predictive
  inference is $\SharpP$-equivalent.\label{pgm1}
\end{theorem}
\begin{proof}
Let $D$ be the set of all denominators appearing in the images of any $\phi_i$. Let $d=\prod_{r\in D} r$. Define new functions $\phi_i'=\phi_i\cdot d$.
Now all input values are integer and so are all weights. This is clearly a problem in $\WCPN$. After the result is obtained, we divide the output
by $d^{|{X}_{\set{J}}|}$ and the result of the original problem is obtained. Hardness is obtained by reducing majority propositional satisfiability (MAJSAT), as previously done in the literature, see for example~\cite{darwiche2009}:
create one $X_i$ for each Boolean variable and a corresponding $\phi_i$ over it
with image equals to 1/2. Create a Boolean circuit representing the propositional formula using auxiliary Boolean variables $X_j$ and functions $\phi_j$, one for
each time that a Boolean operator $\lnot$, $\land$, $\lor$ appears in the propositional formula. Take the last variable $X_n$ which subsumes the circuit output and attach to it the function $\phi_n$ which is an indicator of $X_n=\textrm{true}$. It is easy to see that $Z$ is the proportion of satisfying assignments of the original propositional formula, as the circuit will evaluate to zero for every non-satisfying assignment and to $1/2^N$ for each satisfying assignment (with $N$ the number of Boolean variables of MAJSAT).
\end{proof}

In the case of the generalized probabilistic graphical models, where functions are not given explicitly, we
might end up with a large output size, but we can nevertheless show membership in $\WCPQs$, which is straightforward.

\begin{theorem}
  Given a generalized probabilistic graphical model defined by $\set{P}$, the
  predictive inference is $\WCPQs$-complete. \label{predictioninpq}
\end{theorem}
\begin{proof}
Take the computation paths to correspond to values $x_{\set{J}}\in X_{\set{J}}$ and define the weight $w$ of a path $x_{\set{J}}$ to be the product
$\prod_{i}\phi_i(x_{\set{J}_i})$ of rational functions in $\set{P}$. The sum of rational weights gives exactly the desired value of $Z$, so the problem is in $\WCPQs$.
Hardness comes from the fact that we can use a single $\phi_1$ representing the weights of the $\WCPQs$ problem.
In this case, $Z$ is exactly the sum of weights of the weighted counting problem with rational values.
\end{proof}

Regarding the decision version of {\it predictive inference}, where one queries whether $Z$ is greater than a given rational threshold, the membership in $\PP$ is often
attributed to Littman et al. \cite{littman1998}, where membership of a similar (yet not equal) problem, namely probabilistic acyclic planning, is shown by the construction of a
non-deterministic Turing machine with probability of acceptance greater than half. Such result, if manipulated properly, implies membership for the predictive inference in
probabilistic graphical models too, but it is valid only in cases where the encoding of the instances satisfy some (restrictive) properties. This issue makes that result only
partially satisfactory. Recently, Kwisthout \cite{kwisthout2011} settles the membership of predictive inference in $\PP$ for queried Bayesian networks, but it does not extend
to the generality of probabilistic graphical models as defined here. Hence, we obtain a stronger membership result, because we require only the output size to be polynomially
bounded in the input size. Moreover, this is not restricted to queried Bayesian networks but works for any probabilistic graphical model, including those with functions that
are parametric and shortly encoded (as long as they can be well approximated in polynomial time). In summary, results here lead to a proof that generalizes previous
results for this problem \cite{darwiche2009,kwisthout2011,littman1998,littman2001}.

\begin{theorem}\label{theo:gpgmpp}
  Given a generalized probabilistic graphical model defined by $\set{P}$ such
  that its output size is known to be polynomial in the input size,
the decision version of predictive inference is in $\PP$.
\end{theorem}
\begin{proof}
By applying Theorems \ref{predictioninpq} and \ref{theorem:onecallpp}, the result follows.
\end{proof}

\begin{corollary}\label{cor:pgmpp}
  Given a probabilistic graphical model defined by $\set{P}$,
the decision version of predictive inference is in $\PP$.
\end{corollary}
\begin{proof}
A probabilistic graphical model has explicitly represented functions, so its output size is known to be polynomial in the input size. Hence Theorem~\ref{theo:gpgmpp}
suffices to achieve the desired result.
\end{proof}

Of greater interest is the {\it conditional predictive inference}.
\begin{definition}[Conditional Predictive inference in PGMs]
Given a probabilistic graphical model $\set{P}$ defined by functions $\{ \phi_1, \dotsc, \phi_m \}$ over discrete random variables
$X_{\set{J}}\eq (X_1,\dotsc,X_n)$, {\it conditional predictive inference} (also known as {\it posterior probability computation})
is the task of computing $\Pr({X}_\set{J'}\eq {x}_\set{J'} ~|~ {X}_\set{J''}\eq {x}_\set{J''})$, for given
instantiations ${x}_\set{J'}\in {X}_\set{J'}$ and ${x}_\set{J''}\in {X}_\set{J''}$.
\end{definition}

The computation of conditional predictive inference can be  accomplished by two calls of predictive inference to obtain
$\Pr({X}_\set{J'}\eq {x}_\set{J'}\land {X}_\set{J''}\eq {x}_\set{J''})$ and $\Pr({X}_\set{J''}\eq {x}_\set{J''})$, where probability values $\Pr$ are defined as in
Expression~\eqref{probdef}. Even if the corresponding decision problem for conditional predictive inference is in $\PP$, we are not aware of such proof (numerous papers cite
sources which do not formally prove such result). Hence, a proof is given here.
\begin{theorem}
  Given a generalized probabilistic graphical model defined by $\set{P}$ such
  that its output size is known to be polynomial in the input size,
the decision version of conditional predictive inference is in $\PP$.
\end{theorem}
\begin{proof}
First of all, we can decide whether $\Pr({X}_\set{J''}\eq {x}_\set{J''})>0$ (a call to $\PP$ suffices). If so, then let $q$ be a
rational given as the threshold. We want to decide whether
\begin{align*}
& \frac{\Pr({X}_{\set{J'}\cup\set{J''}}\eq {x}_{\set{J'}\cup\set{J''}})}{\Pr({X}_\set{J''}\eq {x}_\set{J''})} > q \iff\\
  \iff &\Pr({X}_\set{J'}\eq {x}_\set{J'}\land {X}_\set{J''}\eq  {x}_\set{J''}) - q\cdot \Pr({X}_\set{J''}\eq {x}_\set{J''}) > 0\\
  \iff &\sum_{{y}_{\set{J'}\cup\set{J''}}\in {X}_{\set{J'}\cup\set{J''}}} \Pr({X}_\set{J'}\eq {y}_\set{J'}\land {X}_\set{J''}\eq  {y}_\set{J''})
\left[\mathcal{I}(y_{\set{J'}\cup\set{J''}} = {x}_{\set{J'}\cup\set{J''}}) -q\cdot \mathcal{I}(y_{\set{J''}}= {x}_\set{J''})\right] > 0\\
  \iff &\sum_{{y}_{\set{J}}\in {X}_{\set{J}}} \Pr({X}_\set{J}\eq {y}_\set{J})
\left[\mathcal{I}(y_{\set{J'}\cup\set{J''}} = {x}_{\set{J'}\cup\set{J''}}) -q\cdot \mathcal{I}(y_{\set{J''}}= {x}_\set{J''})\right] > 0\\
\iff & \sum_{y_{\set{J}}\in {X}_\set{J}} w(y_{\set{J}}) > 0\, ,
\end{align*}
\noindent where $\mathcal{I}(\cdot)$ is the indicator function. Hence we can see the problem as an instance of $\WDPQs$ with weights $w(y_{\set{J}})=0$ whenever $y_{\set{J}}$ is
not compatible with ${x}_\set{J''}$, $w(y_{\set{J}})=-q\cdot \Pr({X}_\set{J}\eq {y}_\set{J})$ if $y_{\set{J}}$ is compatible with ${x}_\set{J''}$ but not with ${x}_\set{J'}$,
and finally $w(y_{\set{J}})=(1-q)\cdot \Pr({X}_\set{J}\eq {y}_\set{J})$ if $y_{\set{J}}$ is compatible with both ${x}_\set{J'}$ and ${x}_\set{J''}$.
Under the assumptions of this theorem (about bounded output), we can use Theorem~\ref{theorem:onecallpp} to show that this decision is in $\PP$.
Hence we need at most two adaptive calls of $\PP$, which equals to $\PP$ itself.
\end{proof}
\begin{corollary}
  Given a probabilistic graphical model defined by $\set{P}$, the decision version of conditional predictive inference is in $\PP$.
\end{corollary}
\begin{proof}
The same argument as in Corollary~\ref{cor:pgmpp} suffices.
\end{proof}

We have refrained from discussing PP-hardness for these problems, since such result is very well-established, see for instance the book of Darwiche on the topic~\cite{darwiche2009} (the reduction presented in the proof of Theorem~\ref{pgm1} above can be used).

\subsection{Stochastic Combinatorial Optimization}

In this section we present another application of the complexity results, this time in the context of discrete two-stage stochastic combinatorial optimization problems
\cite{stougie1996stochastic}.

Two-stage stochastic combinatorial optimization problems contain uncertainty in terms of stochastic data in their problem formulation. This uncertainty can for example be given
by probability distributions over events of the domain. We call a specific realization of the random events a scenario and we assume that the number of different scenarios is
bounded exponentially in the input size. We further assume that these scenarios can be enumerated efficiently and that the probability that a given scenario occurs can be
computed efficiently in the sense of Definition \ref{definitionWeightedCountingProblem}. In the first stage a decision is made solely based on the given information, without
knowing the actual realizations of the random events. This first-stage decision imposes certain costs. In the second stage, after the realization of the random events is
revealed, another decision has to be taken based on the first-stage decision and on the revealed information. This second-stage decision can for example be used to guarantee
feasibility of the final solution or to react to certain random events. The second-stage decision causes additional costs which are usually called recourse costs. The overall
goal is to find a solution for the given two-stage stochastic combinatorial optimization problem which minimizes the total expected costs, which is defined by the costs of the
first-stage decision plus the expected costs of the second-stage decision.

The formal definition of the model is given below.

\begin{definition}[2-Stage Optimization \cite{stougie1996stochastic}]
We are given a probability distribution over  all possible realizations of the data, called scenarios and denoted by $\mathcal{S}$, and we construct a solution in two stages:
\begin{enumerate}
\item First, we may take some decisions to construct an initial part of the solution, $x$, incurring a cost of $c(x)$.
\item Then some scenario $A \in \mathcal{S}$ is realized according to the distribution, and in the second-stage we may augment the initial decisions by taking recourse
    actions $y(A)$, (if necessary) incurring some cost $f_A(x, y(A))$.
\end{enumerate}

The goal is to choose the initial decisions so as to minimize the expected total cost
\[
c(x)+ E_{A \in \mathcal{S}}[f_A(x, y(A))]\, .
\]
\end{definition}

Dyer and Stougie \cite{dyer2006computational} have shown that discrete two-stage stochastic combinatorial optimization problems are $\SharpP$-hard in general. In order to
obtain this result, they make use of some artificial stochastic combinatorial optimization problems. The result has then been strengthened in
\cite{isco14,weyland2014computational,weylandHardness}, where $\SharpP$-hardness has been shown for a practically relevant stochastic vehicle routing problem. These results are both
imposing lower bounds on the computational complexity of stochastic combinatorial optimization problems. Here we complement them with upper bounds for the computational
complexity of many tasks related to discrete two-stage stochastic combinatorial optimization problems.

In this context the actual solution for a two-stage stochastic combinatorial optimization problems is usually the first-stage decision. This will become more clear with the
following additional assumptions. Assume now that the costs for a first-stage decision can be computed in polynomial time and that for a given first-stage decision and a given
scenario the corresponding recourse costs can also be computed in polynomial time. Given a solution we can compute the expected costs in the following way: By definition we are
able to compute the costs imposed by the first-stage decision in polynomial time. We then enumerate the at most exponentially many scenarios and add to the total costs for each
scenario the recourse costs of this scenario multiplied with the probability that this scenario occurs. Using Theorem \ref{theorem:Approximation} we can prove the following
result for the evaluation of solutions.

\begin{theorem}
We are given a discrete two-stage stochastic combinatorial optimization problem (respecting our assumptions) and a value $b \in \mathbb{N}$. The task of computing the expected
costs for a solution up to an additive error of $2^{-b}$ is $\SharpP$-equivalent.
\end{theorem}

For most of the discrete two-stage stochastic combinatorial optimization problems it holds that the expected costs for any solution can be encoded using at most polynomially
many bits in the input size. In fact, we are not aware of any problem of practical relevance in which this is not the case. Using this additional assumption we can prove the
following results.

\begin{theorem}
\label{theoremUpperBound} We are given a discrete two-stage stochastic combinatorial optimization problem (respecting our extended assumptions). Then the following results
regarding the computational complexity of different computational tasks related to the given problem hold:
\begin{itemize}
 \item[(a)] The task of computing the expected costs for a solution is $\SharpP$-equivalent.
 \item[(b)] The problem of deciding if a given solution has expected costs of at most $t \in \mathbb{Q}$ is in $\PP$.
 \item[(c)] The problem of deciding if a solution with expected costs bounded by $t \in \mathbb{Q}$ exists is in $NP^{\SharpP[1]}$.
 \item[(d)] The problem of computing a solution with minimum expected costs can be solved in polynomial time with access to an ${NP^{\SharpP[1]}}$ oracle.
\end{itemize}
\end{theorem}

These results open some interesting paths for further research. First of all, they describe upper bounds for the complexity of various computational tasks related to two-stage
stochastic combinatorial optimization. In \cite{weylandHardness,weyland2014computational}, it has already been shown that the upper bound of the first result in Theorem
\ref{theoremUpperBound} is tight for a practically relevant problem. It remains to answer whether there are also practically relevant problems whose decision/optimization
variants match the corresponding upper bounds given here.

\section{Conclusions}

We have presented a structured view on weighted counting. We have shown that weighted counting problems are a natural generalization of counting problems and that in many cases
the computational complexity of weighted counting problems corresponds to that of conventional counting problems. The computational complexity of decision problems in $\WDPQs$,
where the size of the output of the associated $\WCPQs$ problem is not necessarily polynomially bounded in the input size, remains an interesting open problem. As for conventional counting problems, it is also
of great interest to improve our understanding of the (polynomial time) approximability and inapproximability of weighted counting problems.

Additionally, we have seen that weighted counting problems arise in many different fields. Using the framework of weighted counting we could give more intuitive
proofs for known results regarding counting problems and quantum computation. We could even obtain new results regarding probabilistic graphical models and two-stage stochastic
combinatorial optimization based on weighted counting. Finally, we believe that our results regarding weighted counting have many more applications in other situations and
fields and we hope that our structured approach to weighted counting might help in revealing such relations in the near future.

\section{Acknowledgments}
The second author's research has been partially supported by the Netherlands Organisation for Scientific Research (NWO) TOP2  (617.001.301) grant and by the \emph{Swiss National Science Foundation} as part of the \emph{Early Postdoc.Mobility} grant P1TIP2\_152282 while the author was at the University of Paris-Dauphine. The third author's research has been partially supported by the \emph{Swiss National Science Foundation} as part of the \emph{Early Postdoc.Mobility} grant P2TIP2\_152293 while the author was at the Department of Economics and Management, University of Brescia, Italy. The authors also thank the reviewers and the editor for very valuable comments and suggestions (including, but not limited to, ideas to prove that PP with approximable rational numbers can decide the halting problem and including the discussion regarding the choice of transition amplitudes in a quantum Turing machine) that greatly improved the manuscript.

\bibliographystyle{elsarticle-num_mine}
\bibliography{article}

\appendix

\section{Rational Approximation of $|x|$}\label{app:a}

In this section we extend the results of~\cite{newman1964rational} to a range of
functions, including polynomials in which all parameters are rational
numbers. By rational parameters we mean polynomials with rational coefficients and integer exponents.
Such extension is straightforward and uses mainly the same proof strategy as in~\cite{newman1964rational}, but to be precise we have rewritten those results as needed. We try to use a
similar notation for the reader who is familiar with~\cite{newman1964rational}, even though we make explicit the dependencies of some variables on each other.
Let $n$ be an integer greater than 4, $c_n\geq 1$, and $\xi_n = e^{-n^{-1/2}}\leq\epsilon_n\leq 1$.
Let $p_n(x), r_n(x)$ be real valued functions defined as follows:

\begin{equation}
p_n(x) = \prod_{k=0}^{n-1} (x+c_n\cdot \epsilon_n^k)\, , ~ ~ ~ r_n(x) = x\cdot\frac{p_n(x) - p_n(-x)}{p_n(x) + p_n(-x)}\, .
\label{eq:px}
\end{equation}

Note that the degree of $r_n(x)$ (defined as the maximum between the degrees of numerator and denominator) is $O(n)$. An intermediate result is useful to prove the next lemma.
\begin{lemma}
Let $n$, $\epsilon_n$ and $\xi_n$ be as defined. Then
$\prod_{j=1}^n \frac{1-\epsilon_n^j}{1+\epsilon_n^j} \leq \xi_n^n$.
\label{newman1}
\end{lemma}
\begin{proof}
 Since $\xi_n\leq\epsilon_n \leq 1$, we have
\[
\prod_{j=1}^n \frac{1-\epsilon_n^j}{1+\epsilon_n^j}
\leq \prod_{j=1}^n \frac{1-\xi_n^j}{1+\xi_n^j} \leq e^{-\sqrt{n}} = \xi_n^n\, .
\]
We have used Lemma 1 of~\cite{newman1964rational} for the last inequality.
\end{proof}
\begin{lemma}
Let $n$, $c_n$, $\epsilon_n$ and $\xi_n$ be as defined. For $c_n\cdot \epsilon_n^n < x\leq c_n$, it holds that $\left|\frac{p_n(-x)}{p_n(x)}\right|\leq \xi_n^n$.
\label{newman2}
\end{lemma}
\begin{proof}
 The proof is very similar to the proof of Lemma 2 in~\cite{newman1964rational}.
Suppose that $c_n\cdot\epsilon_n^{j+1} < x\leq c_n\cdot\epsilon_n^j$, for integer $0\leq j < n$. Then
\begin{eqnarray*}
\left|\frac{p_n(-x)}{p_n(x)}\right| &=& \prod_{k=0}^j \frac{c_n\cdot\epsilon_n^k-x}{c_n\cdot\epsilon_n^k+x}\cdot\prod_{k=j+1}^{n-1} \frac{x-c_n\cdot\epsilon_n^k}{x+c_n\cdot\epsilon_n^k} \leq
\prod_{k=0}^j \frac{c_n\cdot\epsilon_n^k-c_n\cdot\epsilon_n^n}{c_n\cdot\epsilon_n^k+c_n\cdot\epsilon_n^n}\cdot\prod_{k=j+1}^{n-1} \frac{c_n\cdot\epsilon_n^j-c_n\cdot\epsilon_n^k}{c_n\cdot\epsilon_n^j+c_n\cdot\epsilon_n^k} \\
&=& \prod_{m=n-j}^n \frac{1-\epsilon_n^m}{1+\epsilon_n^m}\cdot\prod_{m=1}^{n-j-1}\frac{1-\epsilon_n^m}{1+\epsilon_n^m}=\prod_{m=1}^n\frac{1-\epsilon_n^m}{1+\epsilon_n^m}\, ,
\end{eqnarray*}
\noindent and the result follows from Lemma~\ref{newman1}.
\end{proof}

Theorem~\ref{newmanA} extends the result of~\cite{newman1964rational} in two manners: It allows $x$ to range in a user-defined interval, and it provides a bound based on $\epsilon_n$, which
can be chosen arbitrarily within some bounds. For instance, this will allow us to create fractional polynomials $r_n(x)$ which are defined using only rational parameters (while $\xi_n$ is an irrational number).

\begin{theorem}
Let $n$, $\epsilon_n$ and $\xi_n$ be as defined. Then
$| |x|-r_n(x) | \leq \max\{3\cdot\xi_n^n;c_n\cdot\epsilon_n^n\}$ with $x$ throughout $[-c_n,c_n]$, for $r_n(x)$ constructed as in Expression~\eqref{eq:px}.
\label{newmanA}
\end{theorem}
\begin{proof}
We follow the same arguments as in the proof of Theorem (A) in~\cite{newman1964rational}. Since $|x|$ and $r_n(x)$ are both even, we only need to show that the inequality holds in $[0,c_n]$.
If $0\leq x\leq c_n\cdot\epsilon_n^n$, then the result is trivial, since we have $p_n(-x)\geq 0$ and thus $0\leq r_n(x)\leq x$. Hence
\[
| |x| -r_n(x) | = x-r_n(x) \leq x \leq c_n\cdot\epsilon_n^n\, .
\]
If $c_n\cdot\epsilon_n^n < x\leq c_n$, then
\[
| |x| - r_n(x)| = 2x\cdot \left|\frac{p_n(-x)}{p_n(x)+p_n(-x)}\right| = \frac{2x}{\left|1+\frac{p_n(x)}{p_n(-x)}\right|} \leq \frac{2}{\left|\frac{p_n(x)}{p_n(-x)}\right| - 1}\, ,
\]
\noindent
where the last inequality comes from the fact that $|p_n(x)/p_n(-x)| > 1$ for $c_n\cdot\epsilon_n^n<x\leq c_n$.
Applying Lemma~\ref{newman2} to this last fraction, we have $||x|-r_n(x)| \leq 2/(\xi_n^{-n} -1)$, and since $n > 2$, this is less than $3\cdot\xi_n^n$.
\end{proof}

For instance, if one chooses $\epsilon_n=\xi_n$ and $c_n=3$, then Theorem~\ref{newmanA} extends the result of~\cite{newman1964rational} to the interval $[-3,3]$ without losing any accuracy in the approximation. The following result is useful
when dealing with larger intervals and requiring tight approximations (at the cost of a possibly very large $n$).

\begin{corollary}
If $n=(\nu+f(\nu))^2$ is a perfect square for some integer $\nu\geq 4$, $f$ is
a non-negative integer polynomial, and $\epsilon_n=1-\frac{1}{\sqrt{n}}+\frac{1}{2n}$, then
$| |x|-r_n(x) | < 2^{-\nu}$ holds throughout $[-2^{f(\nu)},2^{f(\nu)}]$ for $r_n(x)$ constructed as in Expression~\eqref{eq:px}. Moreover, $r_n(x)$ is a fraction of polynomials with rational parameters of length polynomial in $\nu$ (and thus in $n$).
\end{corollary}
\begin{proof}
Since $\epsilon_n = 1-\frac{1}{\sqrt{n}}+\frac{1}{2n} =
  1-\frac{1}{\nu+f(\nu)}+\frac{1}{2(\nu+f(\nu))^2}$,
it is clear (by construction) that $r_n(x)$ is a ratio of polynomials with rational parameters of length polynomial in $\nu$.
One can verify that $\epsilon_n$ as just defined satisfies $\xi_n\leq\epsilon_n\leq 1$, because it is the Taylor expansion of $\xi_n$ around zero (that is,
$e^{-1/\sqrt{n}} = \epsilon_n - O(n^{-3/2})$, where $0\leq O(n^{-3/2})\leq\frac{1}{2n}$). Hence, Theorem~\ref{newmanA} implies
$| |x|-r_n(x) | \leq \max\{3\cdot\xi_n^n; 2^{f(\nu)}\cdot \epsilon_n^n\}$ throughout $[-2^{f(\nu)},2^{f(\nu)}]$. Finally,
$3\cdot\xi_n^n = 3\cdot e^{-(\nu+f(\nu))} <  2^{-\nu}$ for $\nu \geq 4$, and
\[
\epsilon_n^n \leq \left(e^{-1/\sqrt{n}} +\frac{1}{2n}\right)^n = e^{-\sqrt{n}}\left(1 + \frac{e^{(1/\sqrt{n}) - \log(2)}}{n}\right)^n < e^{1-\sqrt{n}} = e^{1-(\nu+f(\nu))} < 2^{-(\nu+f(\nu))}\, ,
\]
and so $2^{f(\nu)}\epsilon_n^n < 2^{f(\nu)} 2^{-(\nu+f(\nu))}=2^{-\nu}$ for $\nu \geq 4$.
\end{proof}

\end{document}